\documentclass[10pt,a4paper]{article}

\usepackage{graphicx}
\usepackage{amssymb,amsmath,amsthm}
\usepackage{fullpage}
\usepackage{caption,subcaption}
\usepackage{bbm}
\usepackage{hyperref}

\title{The critical surface fugacity of self-avoiding walks on a rotated honeycomb lattice}
\author{Nicholas R.~Beaton\thanks{nicholas.beaton@lipn.univ-paris13.fr} \\ \ \\ \emph{\small Laboratoire d'Informatique de Paris Nord} \\ \emph{\small Universit\'e Paris 13 (Paris Nord)} \\ \emph{\small 93430 Villetaneuse, France}}
\date{}

\def\x{x_{\rm c}}
\def\e{{\rm e}}
\def\ii{{\rm i}}
\def\y{y_{\rm c}}

\def\SAPPn{{\rm SAPP}_n}

\def\iSAPP{{\rm iSAPP}}

\def\PiSAPP{\mathbb{P}_\iSAPP}

\def\EiSAPP{\mathbb{E}_\iSAPP}

\def\PNiSAPP{\PiSAPP^{\otimes\mathbb{N}}}

\def\PZiSAPP{\PiSAPP^{\otimes\mathbb{Z}}}
\def\bfr{\mathbf{r}}
\def\PPT{P\!P_T}
\def\PPTL{P\!P_{T,L}}

\def\height{\mathsf{H}}
\def\width{\mathsf{W}}
\newtheorem{lem}{Lemma}
\newtheorem{cor}[lem]{Corollary}
\newtheorem{thm}[lem]{Theorem}
\newtheorem{prop}[lem]{Proposition}
\newtheorem{conj}[lem]{Conjecture}

\begin{document}
\maketitle

\begin{abstract}
In a recent paper by Beaton et al, it was proved that a model of self-avoiding walks on the honeycomb lattice, interacting with an impenetrable surface, undergoes an adsorption phase transition when the surface fugacity is $1+\sqrt{2}$. Their proof used a generalisation of an identity obtained by Duminil-Copin and Smirnov, and confirmed a conjecture of Batchelor and Yung. We consider a similar model of self-avoiding walk adsorption on the honeycomb lattice, but with the lattice rotated by $\pi/2$. For this model there also exists a conjecture for the critical surface fugacity, made in 1998 by Batchelor, Bennett-Wood and Owczarek. Using similar methods to Beaton et al, we prove that this is indeed the critical fugacity.
\end{abstract}

\section{Introduction}\label{sec:intro}

Self-avoiding walks (SAWs) have been considered a model of long-chain polymers in solution for a number of decades -- see for example early works by Orr~\cite{Orr1947Statistical} and Flory~\cite{Flory1949Configuration}. In the simplest model one associates a weight (or fugacity) $x$ with each step (or monomer, in the context of polymers) of a walk, and then (for a given lattice) considers the generating function
\[C(x) = \sum_{n\geq0}c_n x^n,\]
where $c_n$ is the number of SAWs starting at a fixed origin and comprising $n$ steps. 

It is straightforward to show (see e.g.~\cite{Madras1993Selfavoiding}) that the limit
\[\mu := \lim_{n\to\infty}c_n^{1/n}\]
exists and is finite. The lattice-dependent value $\mu$ is known as the \emph{growth constant}, and is the reciprocal of the radius of convergence of the generating function $C(x)$. The honeycomb lattice is the only regular lattice in two or more dimensions for which the value of the growth constant is known; its value $\mu=\sqrt{2+\sqrt{2}}$ was conjectured in 1982 by Nienhuis~\cite{Nienhuis1982Exact} and proved by Duminil-Copin and Smirnov in 2012~\cite{DuminilCopin2012Connective}.

The interaction of long-chain polymers with an impenetrable surface can be modelled by restricting SAWs to a half-space, and associating another fugacity $y$ with vertices (or edges) in the boundary of the half-space which are visited by a walk. It is standard practice to place the origin on the boundary. This naturally leads to the definition of a partition function
\[\mathbf C^+_n(y) = \sum_{m\geq0} c^+_n(m) y^m,\]
where $c^+_n(m)$ is the number of $n$-step SAWs starting on the boundary of the half-space and occupying $m$ vertices in the boundary.

The limit
\[\mu(y) := \lim_{n\to\infty}\mathbf C_n^+(y)^{1/n}\]
has been shown to exist for the $d$-dimensional hypercubic lattice for $y>0$ (see e.g.~\cite{Hammersley1982Selfavoiding}). It is a finite, log-convex and non-decreasing function of $y$, and is thus continuous and almost everywhere differentiable. The adaptation of the proof to other regular lattices (in particular, to the honeycomb lattice) is elementary, and will be discussed in Section~\ref{sec:confined_saws}.

It can also be shown that for $0<y\leq1$, 
\[\mu(y) = \mu(1) = \mu,\]
and that $\mu(y)\geq \max\{\mu,y\}$. (This bound is for the hypercubic lattice; as we will see in Section~\ref{sec:confined_saws}, it is slightly different on the honeycomb lattice.) This implies the existence of a \emph{critical fugacity} $\y\geq1$ satisfying
\[\mu(y)\begin{cases}=\mu & \text{if } y\leq\y,\\ >\mu & \text{if }y>\y.\end{cases}\]
This critical fugacity signifies an \emph{adsorption phase transition}, and demarcates the \emph{desorbed} phase $y<\y$ and the adsorbed phase $y>\y$.

Just as the honeycomb lattice is the only regular lattice whose growth constant is known exactly, it is also the only lattice for which an exact value for $\y$ is known. In fact, because there are two different ways to orient the surface (see Figure~\ref{fig:demonstrating_orientations}) for the honeycomb lattice, there are two different values of $\y$. When the surface is oriented so that there are lattice edges perpendicular to the surface (i.e. Figure~\ref{fig:demonstrating_orientations}(a)), the critical fugacity is $\y=1+\sqrt{2}$. This value was conjectured by Batchelor and Yung in 1995~\cite{Batchelor1995Exact}, based on similar Bethe ansatz arguments to those employed by Nienhuis~\cite{Nienhuis1982Exact}, as well as numerical evidence. A proof was discovered by Beaton et al~\cite{Beaton2013Critical}; it used a generalisation of an identity obtained by Duminil-Copin and Smirnov~\cite{DuminilCopin2012Connective}, as well as an adaptation of some results of Duminil-Copin and Hammond~\cite{DuminilCopin2012Selfavoiding}.

\begin{figure}
\centering
\begin{subfigure}{0.4\textwidth}
\includegraphics{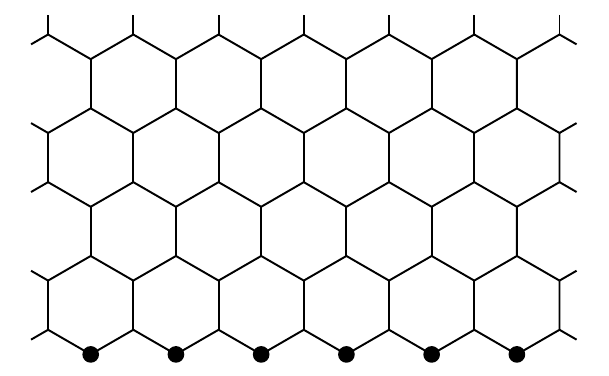}
\caption{}
\end{subfigure}
\hspace{0.5cm}
\begin{subfigure}{0.4\textwidth}
\includegraphics[angle=90]{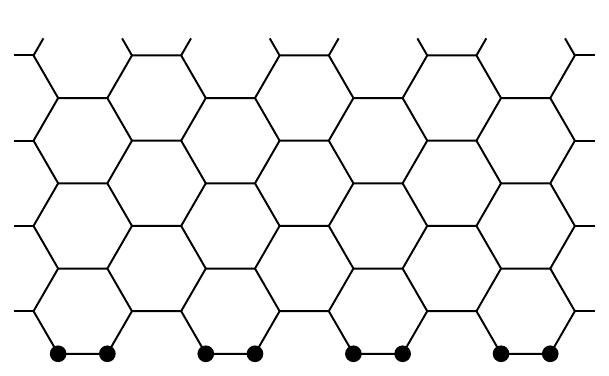}
\caption{}
\end{subfigure}
\caption{The two orientations of an impenetrable surface on the honeycomb lattice, with the surface vertices indicated.}
\label{fig:demonstrating_orientations}
\end{figure}

It is the other orientation of an impenetrable surface on the honeycomb lattice (i.e. Figure~\ref{fig:demonstrating_orientations}(b)) that is the focus of this article. For this model of polymer adsorption there is also a conjecture regarding the critical surface fugacity, due to Batchelor, Bennett-Wood and Owczarek~\cite{Batchelor1998Twodimensional}. In this paper we prove that result:
\begin{thm}\label{thm:main_thm}
For the self-avoiding walk model on the semi-infinite honeycomb lattice with the boundary oriented as per Figure~\ref{fig:demonstrating_orientations}(b), associate a fugacity 
$y$ with occupied vertices on the boundary. Then the model undergoes an adsorption transition at
\[y=y_{\rm c} = \sqrt{\frac{2+\sqrt{2}}{1+\sqrt{2}-\sqrt{2+\sqrt{2}}}} = 2.455\ldots\]
\end{thm}

This paper largely follows the same structure as~\cite{Beaton2013Critical}. We first prove an identity relating several different generating functions of SAWs in a finite domain, evaluated at the critical step fugacity $x=\x=\mu^{-1}$. (We will in fact prove an identity relating generating functions of an $O(n)$ loop model -- a generalisation of the SAW model -- before restricting to the case $n=0$.) We then adapt some existing results for the hypercubic lattice to the honeycomb lattice, and show how the critical fugacity relates to an appropriate limiting case of our identity. This relationship enables us to derive a proof of Theorem~\ref{thm:main_thm}, subject to a certain generating function in a restricted geometry (specifically, the generating function of self-avoiding bridges which span a strip of width $T$) disappearing in a limit. In the appendix we prove that result; the proof is very similar to that of the appendix in~\cite{Beaton2013Critical}, which was in turn based on arguments featured in~\cite{DuminilCopin2012Selfavoiding}.

\subsection{The $O(n)$ model}

While the focus of this paper is the modelling of polymer adsorption with SAWs, some key identities we use are valid for a more general model, called the \emph{$O(n)$ loop model} (closely related to the $n$-vector model). We will thus provide some brief definitions before discussing the identities. 

The $n$-vector model, introduced by Stanley in 1968~\cite{Stanley1968Dependence}, is described by the Hamiltonian
\[\mathcal{H}(d,n) = -J\sum_{\langle i,j\rangle}\mathbf{s}_i\cdot\mathbf{s}_j,\]
where $d$ is the dimensionality of the lattice, $i$ and $j$ are adjacent sites on the lattice, and $\mathbf{s}_i$ is an $n$-dimensional vector of magnitude $\sqrt{n}$. When $n=1$ this is exactly the Ising model, and when $n=2$ it is the classical XY model. De Gennes~\cite{deGennes1972Exponents} showed that it is also possible to take the limit $n\to0$, and that in doing so one obtains the partition function of self-avoiding walks on the lattice.

Of importance to this paper is that it has been shown~\cite{Domany1981Duality} that the $O(n)$ model can be represented by a loop model, where a weight $n$ is associated with each closed loop. The partition function of such a model, on a two-dimensional domain of $N^2$ sites, is written as
\[Z_{N^2}(x) = \sum_\gamma x^{|\gamma|}n^{\ell(\gamma)},\]
where $\gamma$ is a configuration of non-intersecting loops, $|\gamma|$ is the number of edges (or, equivalently, vertices) occupied by $\gamma$ and $\ell(\gamma)$ is the number of closed loops. In this paper we consider a slight generalisation of the model, by allowing a loop configuration to contain a single self-avoiding walk component. (Such an adjustment does not affect the critical point.) When the model is framed this way, it is clear that the case $n\to0$ corresponds to SAWs, as the SAW component of a configuration is the only one with non-zero weight.

Nienhuis's conjecture~\cite{Nienhuis1982Exact} for the growth constant of the SAW model on the honeycomb lattice was in fact a specialisation of a more general result regarding the $O(n)$ model. He predicted that, for $n\in[-2,2]$, the critical point $\x(n)$ is given by
\[\x(n) = \frac{1}{\sqrt{2+\sqrt{2-n}}}.\]
Similarly, the conjecture of Batchelor and Yung~\cite{Batchelor1995Exact} regarding the adsorption of SAWs on the honeycomb lattice (in the orientation of Figure~\ref{fig:demonstrating_orientations}(a)) was also a result for the more general loop model. They predicted that for $n\in[-2,2]$ and $x=1/\sqrt{2+\sqrt{2-n}}$, the critical surface fugacity is given by 
\[\y(n) = 1+\frac{2}{\sqrt{2-n}}.\]

For the rotated orientation that we consider in this paper, the only existing conjecture known to the author for the critical surface fugacity is for the $n=0$ (i.e. self-avoiding walk) model~\cite{Batchelor1998Twodimensional}. We provide the following more general conjecture, based on an identity we obtain and similar behaviour observed in~\cite{Beaton2013Critical}:
\begin{conj}\label{conj:crit_surface_generaln}
For the $O(n)$ loop model on the semi-infinite honeycomb lattice with the boundary oriented as per Figure~\ref{fig:demonstrating_orientations}(b) and with $n\in[-2,2]$, associate a fugacity $\x(n)= 1/\sqrt{2+\sqrt{2-n}}$ with occupied vertices and an additional fugacity $y$ with occupied vertices on the boundary. Then the model undergoes a surface transition at
\[y=\y(n) = \sqrt{\frac{2+\sqrt{2-n}}{1+\sqrt{2-n}-\sqrt{2+\sqrt{2-n}}}}.\]
\end{conj}

\section{The identities}\label{sec:identities}

\subsection{The local identity}\label{ssec:local_ident}

We consider the semi-infinite honeycomb lattice, oriented as in Figure~\ref{fig:demonstrating_orientations}(b), embedded in the complex plane in such a way that the edges have unit length. We refer to north-east/south-west diagonal edges as \emph{positive} and north-west/south-east edges as \emph{negative}.

We follow the examples of~\cite{Beaton2013Critical,DuminilCopin2012Connective} and consider self-avoiding walks which start and end at the \emph{mid-points} (or \emph{mid-edges}) of edges on the lattice. Note that this means the total length of a walk is the same as the number of vertices it occupies. A \emph{vertex-domain} is a collection of vertices $V$ of the lattice with the property that the induced graph of $V$ on the lattice is connected. A collection of mid-edges $\Omega$ is a \emph{domain} if it comprises all the mid-edges adjacent to vertices in a vertex-domain. Equivalently, $\Omega$ is a domain if (a) the set of vertices $V(\Omega)$ adjacent to three mid-edges of $\Omega$ forms a vertex-domain, and (b) every mid-edge of $\Omega$ is adjacent to at least one vertex of $V(\Omega)$.
We denote by $\partial\Omega$ the set of mid-edges of $\Omega$ adjacent to only one vertex of $V(\Omega)$. The \emph{surface} will be a subset of vertices in $V(\Omega)$ which are each adjacent to one mid-edge in $\partial \Omega$. (In practice these will be vertices on the boundary of the half-plane or a strip.)

A \emph{configuration} $\gamma$ consists of a single self-avoiding walk $w$ and a finite collection of closed loops, which are self-avoiding and do not meet another or $w$. We denote by $|\gamma|$ the number of vertices occupied by $\gamma$, by $c(\gamma)$ the number of contacts with the surface (i.e.~vertices on the surface occupied by $\gamma$), and by $\ell(\gamma)$ the number of loops. 

Now define the following so-called \emph{parafermionic observable}: for $a\in\partial\Omega$ and $z\in\Omega$, set
\[F(\Omega,a,p;x,y,n,\sigma) \equiv F(p) := \sum_{\gamma:a\to p} x^{|\gamma|} y^{c(\gamma)} n^{\ell(\gamma)} \e^{-\ii\sigma W(w)},\]
where the sum is over all configurations $\gamma\subset \Omega$ for which the SAW component $w$ runs from $a$ to $p$, and $W(w)$ is the \emph{winding angle} of $w$, that is, $\pi/3$ times the difference between the number of left turns and right turns. See Figure~\ref{fig:example_config} for an example.

\begin{figure}
\centering
\begin{picture}(200,150)
\put(0,0){\includegraphics[height=200pt,angle=90]{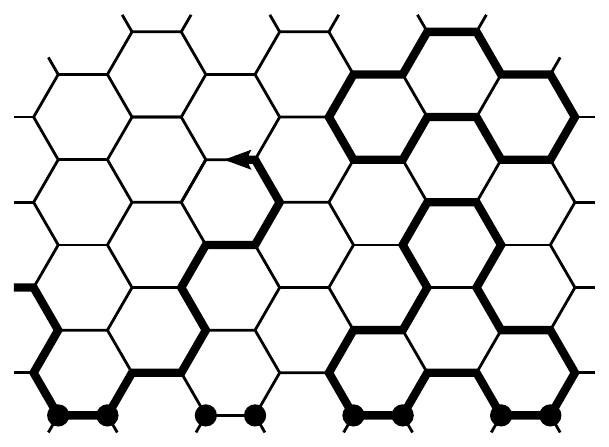}}
\put(-5,50){$a$}
\put(74,100){$z$}
\end{picture}
\caption{A configuration on the honeycomb lattice. The contribution of this configuration to $F(z)$ is $\e^{-\ii\sigma\pi}x^{45}y^6n^2$.}
\label{fig:example_config}
\end{figure}

The following lemma is given as Lemma 3 in~\cite{Beaton2013Critical}; the case $y=1$ is due to Smirnov~\cite{Smirnov2010Discrete}.

\begin{lem}\label{lem:local_ident}
For $n\in[-2,2]$, set $n=2\cos\theta$ with $\theta\in[0,\pi]$. Let
\begin{align}
\sigma&=\frac{\pi-3\theta}{4\pi},\qquad \x^{-1} = 2\cos\left(\frac{\pi+\theta}{4}\right) = \sqrt{2-\sqrt{2-n}},\qquad \text{or} \label{eqn:xcvalue_dense}\\
\sigma&=\frac{\pi+3\theta}{4\pi},\qquad \x^{-1} = 2\cos\left(\frac{\pi-\theta}{4}\right) = \sqrt{2+\sqrt{2-n}}. \label{eqn:xcvalue_dilute}
\end{align}
Then for a vertex $v\in V(\Omega)$ not belonging to the weighted surface, the observable $F$ satisfies
\begin{equation}\label{eqn:local_ident}
(p-v)F(p) + (q-v)F(q) + (r-v)F(r) =0,
\end{equation}
where $p,q,r$ are the three mid-edges adjacent to $v$. If $v$ is in the surface, then let $p$ be the mid-edge adjacent to $v$ in $\partial\Omega$ and take $q,r$ to be the other two mid-edges in counter-clockwise order around $v$. Then
\begin{multline}\label{eqn:local_ident_surface}
  (p-v)F(p) + (q-v)F(q) + (r-v)F(r)= \\
(q-v)(1-y)(\x y\lambda^{\pi/3})^{-1}\sum_{\gamma : a \to q, p} \x ^{|\gamma|} y^{c(
  \gamma)} n^{\ell(\gamma)} \e^{-\ii \sigma W(w)}\\
+(r-v)(1-y) (\x y\lambda^{-\pi/3})^{-1}\sum_{\gamma :  a\to r, p} \x ^{|\gamma|} y^{c(
  \gamma)} n^{\ell(\gamma)} \e^{-\ii \sigma W(w)},
\end{multline}
where 
$\lambda^{\pi/3}=\e^{-\ii\sigma\pi/3}$ is the weight accrued by a walk for
each left turn, and the first
(resp. second) sum runs over configurations $\gamma$ whose SAW component $w$ goes from $a$ to
$p$ via $q$ (resp. via $r$).
\end{lem}

The value of $\x$ given in~\eqref{eqn:xcvalue_dense} corresponds to the larger of the two critical values of the step weight $x$ and hence to the dense regime critical point, while that given in~\eqref{eqn:xcvalue_dilute} corresponds to the line of critical points separating the dense and dilute phases. In what follows we refer to~\eqref{eqn:xcvalue_dense} and~\eqref{eqn:xcvalue_dilute} as the dense and dilute regimes respectively.

\subsection{Self-avoiding walks with no interactions: $n=0$ and $y=1$}\label{ssec:noints}

In~\cite{DuminilCopin2012Connective}, Duminil-Copin and Smirnov use Lemma~\ref{lem:local_ident} to prove that the growth constant of self-avoiding walks ($n=0$ in the dilute regime) is $\x^{-1}= \sqrt{2+\sqrt{2}}$. They do so by considering a special trapezoidal domain, and using the local identity to derive a domain identity satisfied by generating functions of SAWs which end on different sides of the domain. In~\cite{Beaton2013Critical}, the authors generalise that identity to one which relates generating functions of the $O(n)$ loop model and takes into account the surface fugacity $y$.

Here, we construct a similar identity to the one used in~\cite{Beaton2013Critical}. There are several extra complications we must contend with here, and thus we split the derivation into three parts. In this subsection, we consider only the case of SAWs without surface interactions (i.e.~$n=0$ and $y=1$ in the dilute regime). In the next subsections, we will generalise to $n\in[-2,2]$ (in both the dense and dilute regimes) and then to arbitrary $y$. We thus take $\sigma$ and $\x$ to be the values which satisfy Lemma~\ref{lem:local_ident}, that is, $\sigma=5/8$ and $\x = 1/\sqrt{2+\sqrt{2}}$.

We work in the special domain $D_{T,L}$ as illustrated in Figure~\ref{fig:domain_noweights}. (To avoid clutter we have not shown the external half-edges to the left and right of $a$, but they do not play a part in our working. The reasons for this will become clear.) Because we cannot nominate a mid-edge in $\partial D_{T,L}$ as the starting point of SAWs and still preserve any kind of reflective symmetry (which is a key ingredient in the derivation of a useful identity), we instead choose the mid-edge $a$ as the starting point of SAWs. This will introduce some complications around $a$. The height of the domain is the length of the shortest walk starting at $a$ and ending at the top boundary; the width is the number of columns of cells.



\begin{figure}
\centering
\begin{picture}(320,150)
\put(0,0){\includegraphics[height=320pt,angle=90]{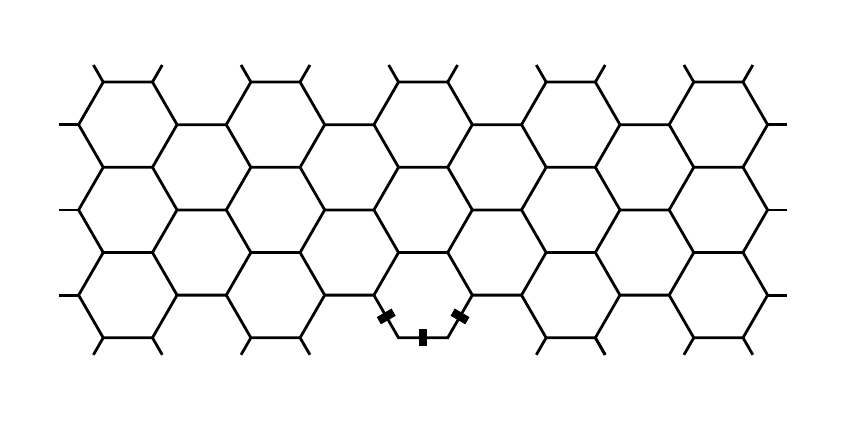}}
\put(300,108){$\epsilon^+$}
\put(300,77){$\epsilon^+$}
\put(300,45){$\epsilon^+$}
\put(5,108){$\epsilon^-$}
\put(5,77){$\epsilon^-$}
\put(5,45){$\epsilon^-$}
\put(217,12){$\alpha^{O+}$}
\put(275,12){$\alpha^{O+}$}
\put(244,12){$\alpha^{I+}$}
\put(78,12){$\alpha^{O-}$}
\put(188,12){$\alpha^{I+}$}
\put(110,12){$\alpha^{I-}$}
\put(54,12){$\alpha^{I-}$}
\put(22,12){$\alpha^{O-}$}
\put(246,138){$\beta^-$}
\put(276,138){$\beta^+$}
\put(221,138){$\beta^+$}
\put(191,138){$\beta^-$}
\put(166,138){$\beta^+$}
\put(136,138){$\beta^-$}
\put(111,138){$\beta^+$}
\put(81,138){$\beta^-$}
\put(56,138){$\beta^+$}
\put(26,138){$\beta^-$}
\put(154,18){$a$}
\put(138,21){$a^-$}
\put(166,21){$a^+$}
\put(176,33){$\zeta^+$}
\put(127,33){$\zeta^-$}
\end{picture}
\caption{The domain $D_{T,L}$ we will use in the proof. Walks start at mid-edge $a$. The labels on the external mid-edges indicate the set containing those mid-edges. This domain has height $T=7$ and width $2L+1=9$.}
\label{fig:domain_noweights}
\end{figure}
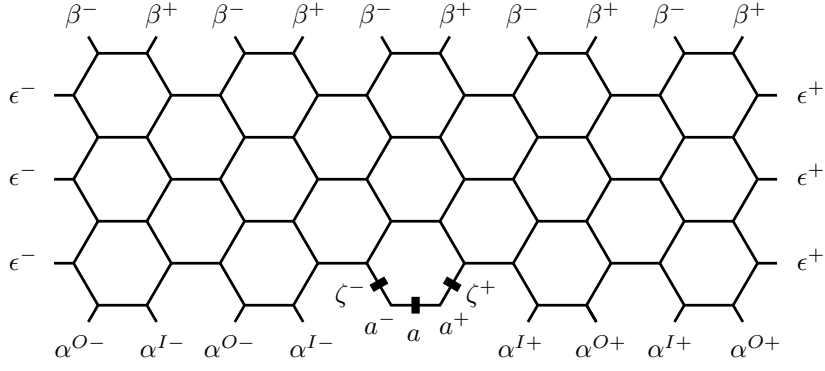



Our identity in this geometry is the following. Define
\begin{align*}
A^O_{T,L}(x) &= \sum_{\gamma:a\to\alpha^{O+}\bigcup\alpha^{O-}}x^{|\gamma|}&
A^I_{T,L}(x) &= \sum_{\gamma:a\to\alpha^{I+}\bigcup\alpha^{I-}}x^{|\gamma|}\\
E_{T,L}(x) &= \sum_{\gamma:a\to\epsilon^+\bigcup\epsilon^-}x^{|\gamma|}&
B_{T,L}(x) &= \sum_{\gamma:a\to\beta^+\bigcup\beta^-}x^{|\gamma|}\\
P_{T,L}(x)&= \sum_{\rho\ni a}x^{|\rho|}
\end{align*}
where the last sum is over all \emph{undirected} (non-empty) self-avoiding polygons in the domain which contain the mid-edge $a$.

\begin{prop}\label{prop:unweighted_identity}
The generating functions $A^O_{T,L},A^I_{T,L},E_{T,L},B_{T,L}$ and $P_{T,L}$, evaluated at $x=\x$, satisfy the identity
\begin{multline} \label{eqn:unweighted_identity}
\sqrt{2-\sqrt{2-\sqrt{2}}}A^O_{T,L}(\x) + \sqrt{2-\sqrt{2+\sqrt{2}}}A^I_{T,L}(\x) + \sqrt{2+\sqrt{2-\sqrt{2}}}E_{T,L}(\x) \\ + \sqrt{2+\sqrt{2+\sqrt{2}}}B_{T,L}(\x) + 2\sqrt{4+2\sqrt{2}-\sqrt{2(10+7\sqrt{2})}}P_{T,L}(\x) \\= \sqrt{8-4\sqrt{2}+2\sqrt{2\left(2-\sqrt{2}\right)}}
\end{multline}
\end{prop}

\begin{proof}
We would like the winding angle of the reflection (through the vertical axis) of a walk to be the negative of that of the original walk. We also require that all walks ending on a given external mid-edge have the same winding angle. Our current definition of the winding angle $W(\gamma)$ cannot accommodate these two requirements; however, this is easily corrected. We define the \emph{adjusted winding angle} $W^*(\gamma)$ to be
\[W^*(\gamma) := \begin{cases} W(\gamma)+\frac{\pi}{2} &\text{if } \gamma \text{ starts in the left direction,} \\ W(\gamma)-\frac{\pi}{2} &\text{if } \gamma \text{ starts in the right direction,} \\ 0 & \text{if } \gamma \text{ is the empty walk}.\end{cases}
\]
(We will henceforth just say winding angle, but this should be interpreted as adjusted winding angle.)


Then, define the \emph{adjusted} observable $F^*$ analogously to $F$:
\[F^*(\Omega,a,p;x,y,n,\sigma) \equiv F^*(p) := \sum_{\gamma:a\to p} x^{|\gamma|} y^{c(\gamma)} n^{\ell(\gamma)} \e^{-\ii\sigma W^*(w)}.\]
Now $F^*$ satisfies the same identity~\eqref{eqn:local_ident} as $F$ for all vertices $v$ \emph{except the two adjacent to $a$}, by exactly the same proof. Because the winding angle of the empty walk is still 0, the identity fails for those two vertices (which we denote $a^-$ and $a^+$ -- see Figure~\ref{fig:domain_noweights}), and they will be treated separately. Thus, we define $V'(D_{T,L}) = V(D_{T,L}) \backslash \{a^-,a^+\}$.

We compute the sum
\begin{equation}\label{eqn:S_sumlocalident}
S = \sum_{\substack{v\in V'(D_{T,L})\\p,q,r\sim v}}(p-v)F^*(p) + (q-v)F^*(q) + (r-v)F^*(r),
\end{equation} 
where $p,q,r$ are the three mid-edges adjacent to vertex $v$, in two ways. 

Firstly, since $F^*$ satisfies the same identity~\eqref{eqn:local_ident} as $F$ for all vertices $v$ in $V'(D_{T,L})$, it follows immediately that $S=0$.

On the other hand, we can compute $S$ by noting that any mid-edge adjacent to two of the vertices summed over will contribute 0 to $S$. The remaining mid-edges are then the external ones marked as in Figure \ref{fig:domain_noweights}, as well as the two adjacent to $a^-$ and $a^+$, which we will denote by $\zeta^-$ and $\zeta^+$ respectively (see Figure~\ref{fig:domain_noweights}). If we define $j=\exp(2\pi\ii/3)/2$ and $\lambda=\exp(-\ii\sigma)=\exp(-5\ii\pi/24)$, then the coefficients of the walks ending on external mid-edges are
\begin{align*}
\alpha^{O+}&: \; -j\lambda^{-5\pi/6} = \frac{\exp(3\pi\ii/16)}{2}&
\alpha^{O-}&: \; \bar{j}\lambda^{5\pi/6} = \frac{\exp(13\pi\ii/16)}{2}\\
\alpha^{I+}&: \; \bar{j}\lambda^{-7\pi/6} = \frac{\exp(\pi\ii/16)}{2}&
\alpha^{I-}&: \; -j\lambda^{7\pi/6} = \frac{\exp(15\pi\ii/16)}{2}\\
\epsilon^+&: \; \frac{\lambda^{-\pi/2}}{2} = \frac{\exp(5\pi\ii/16)}{2}&
\epsilon^-&: \; \frac{-\lambda^{\pi/2}}{2} = \frac{\exp(11\pi\ii/16)}{2}\\
\beta^+&: \; -\bar{j} \lambda^{-\pi/6} = \frac{\exp(7\pi\ii/16)}{2}&
\beta^-&: \; j\lambda^{\pi/6} = \frac{\exp(9\pi\ii/16)}{2}
\end{align*}
There are two types of walks ending at $\zeta^-$ and $\zeta^+$: those with winding $\pm \pi/6$ and those with winding $\mp 7\pi/6$. The first type comprises only one walk for each of $\zeta^-$ and $\zeta^+$: a single step through $a^-$ or $a^+$. The contribution of these walks is thus
\[-j\lambda^{\pi/6}\x + \bar{j}\lambda^{-\pi/6}\x = \frac{-\ii}{2}\sqrt{2-\sqrt{2}+\sqrt{\frac{1}{2} \left(2-\sqrt{2}\right)}}.\]
The second type of walks loop around almost all the way back to $a$: it seems sensible to then just add a step and be left with self-avoiding \emph{polygons} containing the mid-edge $a$. If $P_{T,L}(x)$ is the generating function for \emph{undirected} polygons containing $a$, then the contribution of these walks is
\[\left(\frac{-j\lambda^{7\pi/6}}{\x} + \frac{\bar{j}\lambda^{-7\pi/6}}{\x}\right)P_{T,L}(\x) = \frac{\ii}{2}\sqrt{4+2 \sqrt{2}-\sqrt{2 \left(10+7 \sqrt{2}\right)}}P_{T,L}(\x).\]

For the walks ending on external mid-edges, note that we can pair walks (via reflection through the vertical axis) ending in $\tau^+$ and $\tau^-$, where $\tau$ is any of $\alpha^O, \alpha^I, \epsilon, \beta$. So the contributions of these walks are
\begin{align*}
\alpha^{O+} \cup \alpha^{O-}&: \left(\frac{-j\lambda^{-5\pi/6}+\bar{j}\lambda^{5\pi/6} }{2}\right)A^O_{T,L}(\x) = \frac{\ii}{4}\sqrt{2-\sqrt{2-\sqrt{2}}}A^O_{T,L}(\x)\\
\alpha^{I+} \cup \alpha^{I-}&: \left(\frac{\bar{j}\lambda^{-7\pi/6}-j\lambda^{7\pi/6}}{2}\right)A^I_{T,L}(\x) =  \frac{\ii}{4}\sqrt{2-\sqrt{2+\sqrt{2}}}A^I_{T,L}(\x)\\
\epsilon^+ \cup \epsilon^-&: \left(\frac{\lambda^{-\pi/2}/2-\lambda^{\pi/2}/2 }{2}\right)E_{T,L}(\x) = \frac{\ii}{4}\sqrt{2+\sqrt{2-\sqrt{2}}}E_{T,L}(\x)\\
\beta^+ \cup \beta^-&: \left(\frac{- \bar{j}\lambda^{-\pi/6} +j\lambda^{\pi/6} }{2}\right)B_{T,L}(\x) = \frac{\ii}{4}\sqrt{2+\sqrt{2+\sqrt{2}}}B_{T,L}(\x)
\end{align*}
Adding all the above contributions, equating with 0 and multiplying by $-4\ii$ gives the proposition.
\end{proof}

\subsection{The $O(n)$ model: general $n$}\label{ssec:general_n}

We now generalise Proposition~\ref{prop:unweighted_identity} to allow for configurations which contain closed loops in addition to a SAW component. We take $n\in[-2,2]$, and set $\theta, \sigma$ and $\x$ to be the values which satisfy Lemma~\ref{lem:local_ident}. As before we use $\lambda = \exp(-\ii\sigma)$. 

The generating functions $A^O_{T,L}, A^I_{T,L}, E_{T,L}, B_{T,L}$ now count configurations which may contain non-intersecting closed loops as well as the SAW component. They thus now have a second argument $n$, which is conjugate to the number of closed loops. The function $P_{T,L}$ now counts configurations with a closed loop containing $a$ and (possibly) other non-intersecting closed loops. However, things become tricky if we associate the weight $n$ with the closed loop containing $a$, and so we define
\[P_{T,L}(x;n) =\sum_{\rho\ni a}x^{|\rho|}n^{\ell(\rho)-1}.\]
We introduce another generating function,
\[G_{T,L}(x;n) = \sum_{\rho\not\ni a}x^{|\rho|} n^{\ell(\rho)},\]
which counts the empty configuration plus configurations of closed loops which \emph{do not} contain $a$.

\begin{prop}\label{prop:identity_generaln}
If $n=2\cos\theta$ with $\theta\in[0,\pi]$ and $\x^{-1} = 2\cos((\pi\pm\theta)/4)$, then
\begin{multline}\label{eqn:ident_generaln}
\cos\left(\frac{5(\pi\pm\theta)}{8}\right)A^O_{T,L}(\x;n) + \cos\left(\frac{9\pi\mp7\theta}{8}\right)A^I_{T,L}(\x;n)\\
+\cos\left(\frac{3(\pi\pm\theta)}{8}\right)E_{T,L}(\x;n) + \frac{2}{\x}\cos\left(\frac{7(\pi\pm\theta)}{8}\right)P_{T,L}(\x;n)\\
+\cos\left(\frac{\pi\pm\theta}{8}\right)B_{T,L}(\x;n) = 2\x\cos\left(\frac{\pi\pm\theta}{8}\right)G_{T,L}(\x;n).
\end{multline}
\end{prop}

\begin{proof}
We proceed in exactly the same way as the SAW case. Take $\theta, \sigma$ and $\x$ to be the values which satisfy Lemma~\ref{lem:local_ident}. We once again consider $S$ (see~\eqref{eqn:S_sumlocalident}), which by Lemma~\ref{lem:local_ident} is still 0.

As with Proposition~\ref{prop:unweighted_identity}, we can also compute $S$ by noting that any mid-edge adjacent to two vertices being summed over will contribute 0. This leaves the mid-edges in the boundary of $D_{T,L}$ as well as $\zeta^-$ and $\zeta^+$. We can use the reflective symmetry of the domain to pair walks ending on the boundary. (Once again we set $j=\exp(2\pi\ii/3)/2$ and $\lambda = \exp(-\ii\sigma)$.)
\begin{align*}
\alpha^{O+} \cup \alpha^{O-}&: \left(\frac{-j\lambda^{-5\pi/6}+\bar{j}\lambda^{5\pi/6} }{2}\right)A^O_{T,L}(\x;n) = \frac{\ii}{2}\cos\left(\frac{5(\pi\pm\theta)}{8}\right)A^O_{T,L}(\x;n)\\
\alpha^{I+} \cup \alpha^{I-}&: \left(\frac{\bar{j}\lambda^{-7\pi/6}-j\lambda^{7\pi/6}}{2}\right)A^I_{T,L}(\x;n) =  \frac{\ii}{2}\cos\left(\frac{9\pi\mp7\theta}{8}\right)A^I_{T,L}(\x;n)\\
\epsilon^+ \cup \epsilon^-&: \left(\frac{\lambda^{-\pi/2}/2-\lambda^{\pi/2}/2 }{2}\right)E_{T,L}(\x;n) = \frac{\ii}{2}\cos\left(\frac{3(\pi\pm\theta)}{8}\right)E_{T,L}(\x;n)\\
\beta^+ \cup \beta^-&: \left(\frac{- \bar{j}\lambda^{-\pi/6} +j\lambda^{\pi/6} }{2}\right)B_{T,L}(\x;n) = \frac{\ii}{2}\cos\left(\frac{\pi\pm\theta}{8}\right)B_{T,L}(\x;n)
\end{align*}

For the configurations whose SAW component ends at $\zeta^-$ or $\zeta^+$, we again consider separately those with winding $\pm\pi/6$ and those with winding $\mp7\pi/6$. For the former, the SAW component of a configuration is still just a single step to $\zeta^-$ or $\zeta^+$. The loop component can thus be any collection of loops which do not contain $a$. The contribution of these configurations is
\[\left(-j\lambda^{\pi/6}\x + \bar{j}\lambda^{-\pi/6}\x\right)G_{T,L}(\x;n) = -\ii\x\cos\left(\frac{\pi\pm\theta}{8}\right)G_{T,L}(\x;n).\]

For the second type of configurations ending at $\zeta^-$ or $\zeta^+$, we can again add a step and view the SAW component as a loop which contains $a$. However, this loop does not naturally contribute a factor of $n$, and so for these configurations $n$ will be conjugate to one less than the number of loops. (Hence the definition of $P_{T,L}(x;n)$.) These configurations thus contribute
\[\left(\frac{-j\lambda^{7\pi/6}}{\x} + \frac{\bar{j}\lambda^{-7\pi/6}}{\x}\right)P_{T,L}(\x;n) = \frac{\ii}{\x}\cos\left(\frac{7(\pi\pm\theta)}{8}\right)P_{T,L}(\x;n).\]
Adding all the above contributions together, equating with 0 and multiplying by $-2\ii$ gives the proposition.
\end{proof}

\subsection{Including surface interactions: general $y$}\label{ssec:general_y}
We now wish to introduce surface weights. As was the case in~\cite{Beaton2013Critical}, we associate the surface fugacity $y$ with vertices on the $\beta$ boundary of $D_{T,L}$. (We could derive an identity with the weights on the $\alpha$ boundary, but some of the coefficients would be negative, and this would prevent us from completing a proof of the critical fugacity.) See Figure~\ref{fig:domain_withweights} for an illustration. With the surface weights, we will require $T+L\equiv 1\,(\text{mod }2)$, as we will need to pair together SAWs which end at $\beta^-$ and $\beta^+$ mid-edges.

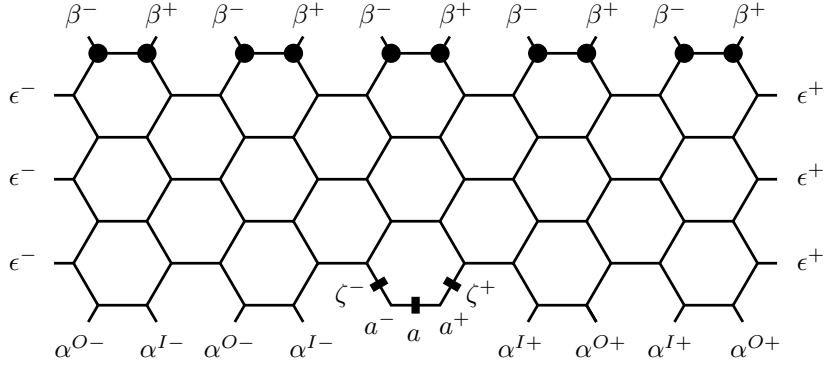
\begin{figure}
\centering
\begin{picture}(320,150)
\put(0,0){\includegraphics[height=320pt,angle=90]{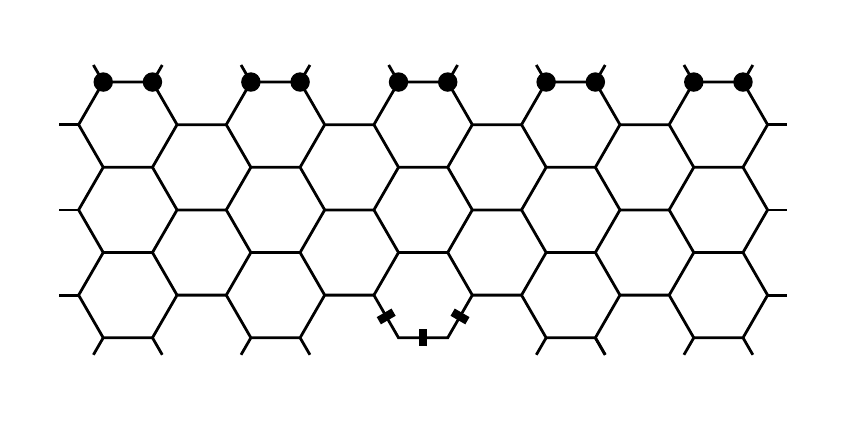}}
\put(300,108){$\epsilon^+$}
\put(300,77){$\epsilon^+$}
\put(300,45){$\epsilon^+$}
\put(5,108){$\epsilon^-$}
\put(5,77){$\epsilon^-$}
\put(5,45){$\epsilon^-$}
\put(217,12){$\alpha^{O+}$}
\put(275,12){$\alpha^{O+}$}
\put(244,12){$\alpha^{I+}$}
\put(78,12){$\alpha^{O-}$}
\put(188,12){$\alpha^{I+}$}
\put(110,12){$\alpha^{I-}$}
\put(54,12){$\alpha^{I-}$}
\put(22,12){$\alpha^{O-}$}
\put(246,138){$\beta^-$}
\put(276,138){$\beta^+$}
\put(221,138){$\beta^+$}
\put(191,138){$\beta^-$}
\put(166,138){$\beta^+$}
\put(136,138){$\beta^-$}
\put(111,138){$\beta^+$}
\put(81,138){$\beta^-$}
\put(56,138){$\beta^+$}
\put(26,138){$\beta^-$}
\put(154,18){$a$}
\put(138,21){$a^-$}
\put(166,21){$a^+$}
\put(176,33){$\zeta^+$}
\put(127,33){$\zeta^-$}
\end{picture}
\caption{The domain $D_{7,4}$, with the weighted vertices on the $\beta$ boundary indicated. For the surface-weighted case, we require $T+L\equiv 1\, (\text{mod } 2)$.}
\label{fig:domain_withweights}
\end{figure}

We use the same generating functions as in Proposition~\ref{prop:identity_generaln}, but now with a variable $y$ keeping track of the number of surface contacts. For example,

\[A^O_{T,L}(x,y;n) = \sum_{\gamma:a\to\alpha^{O+}\bigcup\alpha^{O-}}x^{|\gamma|}n^{\ell(\gamma)}y^{c(\gamma)}.\]



\begin{prop}\label{prop:identity_general_ny}
If $n=2\cos\theta$ with $\theta\in[0,\pi]$, $\x^{-1}=2\cos((\pi\pm\theta)/4)$ and $T+L\equiv1\,(\text{mod }2)$, then
\begin{multline}\label{eqn:ident_general_ny}
\cos\left(\frac{5(\pi\pm\theta)}{8}\right)A^O_{T,L}(\x,y;n) + \cos\left(\frac{9\pi\mp7\theta}{8}\right)A^I_{T,L}(\x,y;n)\\
+\cos\left(\frac{3(\pi\pm\theta)}{8}\right)E_{T,L}(\x,y;n) + \frac{2}{\x}\cos\left(\frac{7(\pi\pm\theta)}{8}\right)P_{T,L}(\x,y;n)\\
+\left[\cos\left(\frac{\pi\pm\theta}{8}\right)-\frac{(1-\x y-\x^2 y^2)\cos\left(\frac{9\pi\pm\theta}{8}\right) + \x^2 y^2\cos\left(\frac{5(\pi\pm\theta)}{8}\right)}{\x y(1+\x y)}\right]B_{T,L}(\x,y;n) \\= 2\x\cos\left(\frac{\pi\pm\theta}{8}\right)G_{T,L}(\x,y;n).
\end{multline}
\end{prop}

\begin{proof}
We again consider the sum $S$. When $y\neq 1$ the contribution of the weighted $\beta$ vertices will not be 0, but the total can instead be written as a multiple of the $B_{T,L}$ generating function.

A $\beta^-$ configuration must approach its final vertex either from the south-west or from the east. Let $\gamma_1$ be a configuration approaching a $\beta^-$ vertex from the south-west, and say $\gamma_1^l$ and $\gamma_1^r$ are the configurations obtained by appending a left or right turn to $\gamma_1$ respectively. Then the sum of the contributions of $\gamma_1, \gamma_1^l$ and $\gamma_1^r$ is
\[\x^{|\gamma_1|}n^{\ell(\gamma_1)}y^{c(\gamma_1)}(\bar{j} \lambda^{-\pi/6}+\x y j\lambda^{\pi/6}+\x y\lambda^{-\pi/2}).\]
Similarly, let $\gamma_2$ be a configuration approaching a $\beta^-$ vertex from the east, and $\gamma_2^l$ and $\gamma_2^r$ its two extensions. Then the contribution of these three walks is
\[\x^{|\gamma_2|}n^{\ell(\gamma_2)}y^{c(\gamma_2)}(\lambda^{\pi/2}+\x y \bar{j} \lambda^{5\pi/6}+\x y j\lambda^{\pi/6})\]
Now any configuration finishing adjacent to a $\beta^-$ vertex must be described by exactly one of $\gamma_1,\gamma_1^l,\gamma_1^r,\gamma_2,\gamma_2^l,\gamma_2^r$. So if $\Gamma^1_{T,L}(x,y;n)$ and $\Gamma^2_{T,L}(x,y;n)$ are the generating functions for $\gamma_1$ and $\gamma_2$ configurations respectively, the contribution of all $\beta^-$ vertices is
\begin{equation}\label{eqn:betaplus_contr}
(\bar{j} \lambda^{-\pi/6}+\x y j \lambda^{\pi/6}+\x y \lambda^{-\pi/2})\Gamma^1_{T,L}(\x,y;n) + (\lambda^{\pi/2}+\x y \bar{j} \lambda^{5\pi/6}+\x y j \lambda^{\pi/6})\Gamma^2_{T,L}(\x,y;n).
\end{equation}
But now it's easy to see that any reflected (in the vertical axis) $\gamma_1$ walk can be extended to a unique $\gamma_2$ walk, and any $\gamma_2$ walk is an extension of a unique reflected $\gamma_1$ walk. So in fact
\[\Gamma^2_{T,L}(x,y;n) = xy\Gamma^1_{T,L}(x,y;n).\]
So \eqref{eqn:betaplus_contr} becomes
\begin{equation}\label{eqn:betaplus_simplified}
(\bar{j} \lambda^{-\pi/6}+\x y j \lambda^{\pi/6}+\x y \lambda^{-\pi/2}+\x y \lambda^{\pi/2}+\x^2 y^2\bar{j}\lambda^{5\pi/6} + \x^2y^2j\lambda^{\pi/6})\Gamma^1_{T,L}(\x,y;n).
\end{equation}
Since any $\beta^-$ vertex can be reflected through the vertical axis to give a $\beta^+$ vertex, the contribution of $\beta^+$ vertices is
\begin{equation}\label{eqn:betaminus_simplified}
(-j\lambda^{\pi/6} - \x y\bar{j}\lambda^{-\pi/6} - \x y\lambda^{\pi/2}-\x y\lambda^{-\pi/2}-\x^2 y^2j\lambda^{-5\pi/6} - \x^2y^2\bar{j}\lambda^{-\pi/6})\Gamma^1_{T,L}(\x,y;n).
\end{equation}
So the contribution of all $\beta^-$ and $\beta^+$ vertices is
\begin{align}
&\left[(1-\x y-\x^2 y^2)(\bar{j}\lambda^{-\pi/6}-j\lambda^{\pi/6})+\x^2 y^2(\bar{j}\lambda^{5\pi/6}-j\lambda^{-5\pi/6})\right]\Gamma^1_{T,L}(\x,y;n)\notag\\
&= \ii\left[(1-\x y-\x^2 y^2)\cos\left(\frac{9\pi\pm\theta}{8}\right) + \x^2 y^2\cos\left(\frac{5(\pi\pm\theta)}{8}\right)\right]\Gamma^1_{T,L}(\x,y;n).\label{eqn:contribution_B_verts}
\end{align}

Now any walk counted by $B_{T,L}$ can be obtained by extending a unique $\Gamma^1_{T,L}$ walk (or a reflected one) by either a single step or by two steps. Similarly, any $\Gamma^1_{T,L}$ walk (or a reflected one) can be extended by one or two steps to give a $B_{T,L}$ walk. So we have
\begin{equation}\label{eqn:BintermsofG}
B_{T,L}(x,y;n) = 2(xy +x^2y^2)\Gamma^1_{T,L}(x,y;n).
\end{equation}
Combining~\eqref{eqn:contribution_B_verts} and~\eqref{eqn:BintermsofG}, we can write the contribution of all $\beta^-$ and $\beta^+$ vertices in terms of $B_{T,L}(\x,y;n)$. We thus find
\begin{equation}\label{eqn:S_weighted_itoB}
S=\frac{\ii}{2\x y(1+\x y)}\left[(1-\x y-\x^2 y^2)\cos\left(\frac{9\pi\pm\theta}{8}\right) + \x^2 y^2\cos\left(\frac{5(\pi\pm\theta)}{8}\right)\right]B_{T,L}(\x,y;n).
\end{equation}

The second method for calculating $S$, by noting that internal mid-edges contribute 0 to the sum, does not change from the unweighted case, and so we have
\begin{multline}\label{eqn:ident_general_ny_unarranged}
\frac{\ii}{2}\cos\left(\frac{5(\pi\pm\theta)}{8}\right)A^O_{T,L}(\x,y;n) + \frac{\ii}{2}\cos\left(\frac{9\pi\mp7\theta}{8}\right)A^I_{T,L}(\x,y;n)\\
+\frac{\ii}{2}\cos\left(\frac{3(\pi\pm\theta)}{8}\right)E_{T,L}(\x,y;n) + \frac{\ii}{\x}\cos\left(\frac{7(\pi\pm\theta)}{8}\right)P_{T,L}(\x,y;n)\\
+\frac{\ii}{2}\cos\left(\frac{\pi\pm\theta}{8}\right)B_{T,L}(\x,y;n) - \ii\x\cos\left(\frac{\pi\pm\theta}{8}\right)G_{T,L}(\x,y;n)\\
=\frac{\ii}{2\x y(1+\x y)}\left[(1-\x y-\x^2 y^2)\cos\left(\frac{9\pi\pm\theta}{8}\right) + \x^2 y^2\cos\left(\frac{5(\pi\pm\theta)}{8}\right)\right]B_{T,L}(\x,y;n).
\end{multline}
Multiplying by $-2\ii$ and rearranging gives the result of the proposition.
\end{proof}

We mention here that the coefficient of $B_{T,L}(\x,y;n)$ in the dilute version of~\eqref{eqn:ident_general_ny} is 0 when
\begin{align*}
y&=2\sin\left(\frac{\pi+\theta}{4}\right)\sqrt{\frac{\sin\left(\frac{-5\pi+\theta}{8}\right)}{\sin\left(\frac{\pi-5\theta}{8}\right)}}\\
&= \sqrt{\frac{2+\sqrt{2-n}}{1 +\sqrt{2-n}-\sqrt{2+\sqrt{2-n}}}}.
\end{align*}
Since we are able to prove that when $n=0$ this is the critical surface fugacity for SAWs, we conjecture that for $n\in[-2,2]$ this is the critical surface fugacity for the $O(n)$ loop model (stated earlier as Conjecture~\ref{conj:crit_surface_generaln}). We are, however, unable to prove this for general values of $n$.

\section{Confined self-avoiding walks}\label{sec:confined_saws}

We now consider only SAWs rather than the more general $O(n)$ model; that is, we will specialise to $n=0$ in the dilute regime. It is for this model that we are able to derive a proof of the critical surface fugacity $\y$ (Theorem~\ref{thm:main_thm}). The identity~\eqref{eqn:ident_generaln} will form a crucial part of this proof, but we first need to establish some results which relate $\y$ to the generating functions featured in~\eqref{eqn:ident_generaln}. 
By necessity this section is essentially the same as Section 3 of~\cite{Beaton2013Critical}, and thus for several proofs we will refer the reader to that article.

\subsection{Self-avoiding walks in a half-plane}\label{ssec:half_plane}

Recall from Section~\ref{sec:intro} that we define the partition function
\[\mathbf C_n^+(y) = \sum_{m\geq0}c^+_n(m)y^m,\]
where $c^+_n(m)$ is the number of $n$-step SAWs starting on the boundary of the half-space and occupying $m$ vertices in the surface. In keeping with the methodology of the previous section, we will consider SAWs to start and end on mid-edges of the lattice. As we did in the previous section, we will take SAWs to begin on the mid-edge of a horizontal edge lying along the surface. We will take $(\gamma_o,\gamma_1,\ldots,\gamma_n)$ to be the sequence of mid-edges defining a SAW $\gamma$.
\begin{prop}\label{prop:FE_etc_existence}
For $y>0$, 
\[\mu(y):=\lim_{n\to\infty}\mathbf C_n^+(y)^{1/n}\]
exists and is finite. It is a log-convex, non-decreasing function of $\log y$, and therefore continuous and almost everywhere differentiable.

For $0<y\leq 1$,
\[\mu(y) = \mu(1) = \mu,\]
where $\mu=\sqrt{2+\sqrt{2}}$ is the growth constant of SAWs on the honeycomb lattice. Moreover, for any $y>0$, 
\[\mu(y)\geq\max\{\mu,\sqrt{y}\}.\]
This implies the existence of a critical value $\y$, with $1\leq \y\leq\mu^2$, which delineates the transition from the desorbed phase to the adsorbed phase:
\[\mu(y)\begin{cases}=\mu &\text{if }y\leq\y,\\>\mu& \text{if }y>\y.\end{cases}\]
\end{prop}
The existence of $\mu(y)$ has been proved by Hammersley, Torrie and Whittington~\cite{Hammersley1982Selfavoiding} in the case of the $d$-dimensional hypercubic lattice. Their proof uses a type of SAW called an \emph{unfolded walk}, which is a SAW whose origin and end-point have minimal and maximal $\mathbf x$-coordinates respectively. The usefulness of unfolded walks arises from the fact that they can be concatenated freely without creating self-intersections. 
If $u^+_n(m)$ and $\mathbf U^+_n(y)$ are defined for unfolded walks analogously to $c^+_n(m)$ and $\mathbf C^+_n(y)$, then it is straightforward to show
\[\lim_{n\to\infty}\mathbf U^+_n(y)^{1/n}\]
exists and satisfies analogous properties to those described in the proposition.

To relate unfolded walks to general SAWs, Hammersley, Torrie and Whittington show that for the hypercubic lattice,
\begin{equation}\label{eqn:unfolded_same_limit}
\lim_{n\to\infty}\mathbf C^+_n(y)^{1/n} = \lim_{n\to\infty}\mathbf U^+_n(y)^{1/n}.
\end{equation}
They use a process called \emph{unfolding} to relate regular and unfolded SAWs. Unfolding consists of reflecting parts of a walk through lines parallel to the $\mathbf y$-axis and passing through vertices of the walk with maximal or minimal $\mathbf x$-coordinates, until the resulting walk is unfolded. The number of SAWs which result in the same unfolded walk can be bounded above by a sub-exponential term, to result in the inequality
\begin{equation}\label{eqn:ineq_unfolding}
e^{-c\sqrt{n}} \mathbf C_n^+(y) \leq (1+1/y)\mathbf U_{n+1}^+(y) \leq (1+1/y)\mathbf C_{n+1}^+(y),\qquad n\geq N_c,
\end{equation}
where $c$ and $N_c$ are positive (finite) constants. (The difference in lengths arises because the authors always add a step at the end when unfolding, to guarantee that the endpoint is strictly to the right of all other points.) Raising all terms to the power of $1/n$ and taking appropriate lim infs and lim sups yields~\eqref{eqn:unfolded_same_limit}.




\begin{proof} 
The concept of an unfolded walk is well-defined on the honeycomb lattice: let $u^+_n(m)$ be the number of $n$-step walks in the upper half-plane which start at a horizontal mid-edge on the surface, visit $m$ surface vertices and whose starting (resp. ending) point has minimal (resp. maximal) $\mathbf x$-coordinate. Then, let 
\[\mathbf U^+_n(y) := \sum_{m\geq 0}u^+_n(m)y^m.\]
For unfolded walks on the honeycomb lattice, the proof of~\cite{Hammersley1982Selfavoiding} can be applied \emph{mutatis mutandis} to show that 
\[\mu(y) = \lim_{n\to\infty} \mathbf U_n^+(y)^{1/n}\]
exists and satisfies the properties given in the proposition. (Note that when concatenating unfolded walks on the honeycomb lattice, the addition of one or two mid-edges at the point of concatenation may be necessary. This does not interfere with the proof.)

The process of unfolding on our honeycomb lattice is made more complicated by the fact that the lattice is not invariant under reflection through a vertical line passing through a vertex. It is thus necessary to insert a horizontal edge into a walk each time we reflect a component. The number of edges added when unfolding a walk of length $n$ is at most $O(\sqrt{n})$, which follows from the computation of the maximum number of pieces in a partition of $n$ when all pieces are distinct. Indeed, the exact value of this maximum is\footnote{This is simply the floor of the inverse of the triangular number function, $T(x) = x(x+1)/2$.}
\[\Bigg\lfloor\frac{-1+\sqrt{1+8n}}{2}\Bigg\rfloor.\] 
Working backwards, one observes that while the \emph{number} of SAWs which result in the same unfolded walk is at most $e^{c\sqrt{n}}$ for some $c>0$, the \emph{lengths} of those walks can range in $[n-\delta\sqrt{n},n]$ for some $\delta>0$.

This is undesirable, as it prevents us from writing a simple relation like~\eqref{eqn:ineq_unfolding}. Instead, we define a process called \emph{fixed-length unfolding}. The procedure, applied to a walk $\gamma$ of length $n$, is simply the following:
\begin{enumerate}
\item If the starting point of $\gamma$ already satisfies $\mathbf x(\gamma_0) < \mathbf x (\gamma_i)$ for $1\leq i\leq n$, skip to step 3. Otherwise, let $\mathcal{L}$ be the set of vertices visited by $\gamma$ satisfying $v\in\mathcal{L} \Rightarrow \mathbf x(v) \leq \mathbf x(\gamma_i)$ for $0\leq i \leq n$, and let $v_0$ be the first of those vertices visited by $\gamma$. Let $p(\gamma)$ be the section of $\gamma$ from $\gamma_0$ to $v_0$, and $s(\gamma)$ be the section of $\gamma$ from $v_0$ to $\gamma_n$. 
\item Take $\gamma \mapsto r(p(\gamma))\circ \mathbbm h \circ s(\gamma)$, where $r$ denotes reflection through the vertical axis, $\mathbbm h$ is a horizontal edge, and $\circ$ denotes concatenation. Return to step 1.
\item If the endpoint already satisfies $\mathbf x(\gamma_n) > \mathbf x(\gamma_i)$ for $0\leq i \leq n-1$, skip to step 5 (a). Otherwise, let $\mathcal{R}$ be the set of vertices visited by $\gamma$ satisfying $v\in\mathcal{R} \Rightarrow \mathbf x(v) \geq \mathbf x(\gamma_i)$ for $0\leq i \leq n$, and let $v_1$ be the last of those vertices visited by $\gamma$. Let $p(\gamma)$ be the section of $\gamma$ from $\gamma_0$ to $v_1$, and $s(\gamma)$ be the section of $\gamma$ from $v_1$ to $\gamma_n$. 
\item Take $\gamma \mapsto p(\gamma)\circ \mathbbm h \circ r(s(\gamma))$. Return to step 3.
\item \begin{enumerate}
  \item If $|\gamma|\geq n+\lfloor 3\sqrt{n}\rfloor$ we are done. Otherwise, take $\gamma \mapsto \gamma \circ \mathbbm w(\gamma)$, where 
\[\mathbbm w(\gamma) = \begin{cases} 
\mathbbm l \circ \mathbbm r \circ \mathbbm r \circ \mathbbm l & \text{if }\gamma \text{ ends on a horizontal edge,}\\
\mathbbm l \circ \mathbbm l \circ \mathbbm r \circ \mathbbm r & \text{if }\gamma \text{ ends on a negative edge,}\\
\mathbbm r \circ \mathbbm r \circ \mathbbm l \circ \mathbbm l & \text{if }\gamma \text{ ends on a positive edge,}
\end{cases}\]
 and $\mathbbm l$ (resp. $\mathbbm r$) is a single left (resp. right) turn. Skip to step 6.
\item If $|\gamma|\geq n+\lfloor 3\sqrt{n}\rfloor$ we are done. Otherwise, define $q_i(\gamma)$ to be the walk resulting from inserting $\mathbbm w = \mathbbm l \circ \mathbbm r \circ \mathbbm r \circ \mathbbm l$ into the $i$-th position of $\gamma$, and take
\[\gamma \mapsto \begin{cases} 
q_{|\gamma|-2}(\gamma) & \text{if }\gamma \text{ ends on a horizontal edge,}\\
q_{|\gamma|-1}(\gamma) & \text{if }\gamma \text{ ends on a negative edge,}\\
q_{|\gamma|-3}(\gamma) & \text{if }\gamma \text{ ends on a positive edge.}\\
\end{cases}\]
  \end{enumerate}
\item If the endpoint of $\gamma$ is adjacent to a surface vertex, go to step 5 (b). Otherwise, return to step 5 (a).
\end{enumerate}
In Figure~\ref{fig:unfolding} we illustrate a walk before and after fixed-length unfolding.

\begin{figure}
\centering
\begin{picture}(450,200)
\put(0,0){\includegraphics[width=450pt]{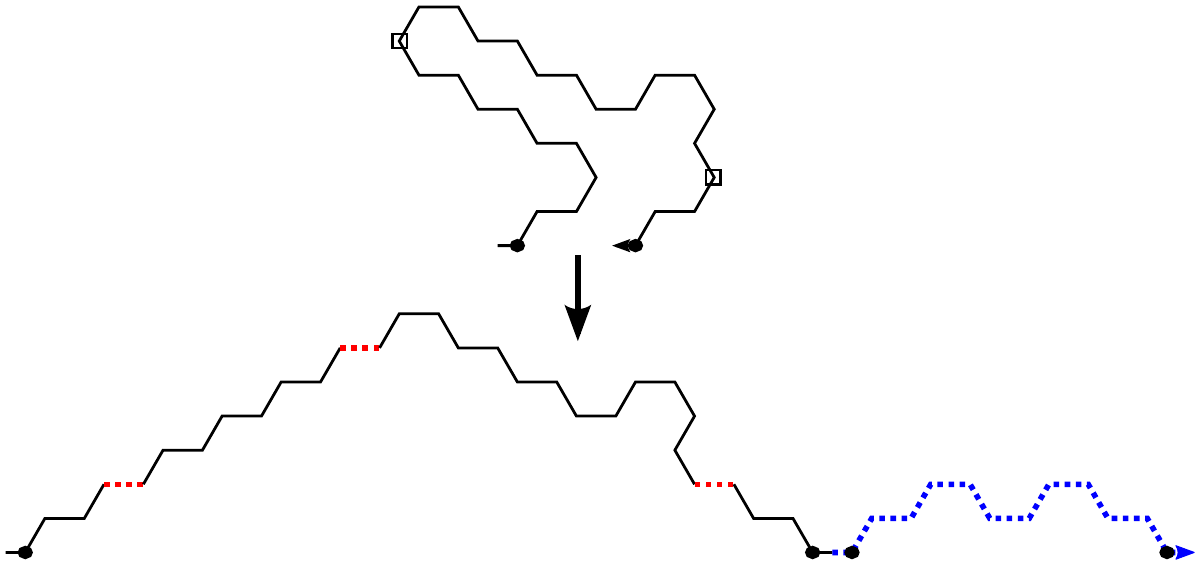}}
\put(134,195){$v_0$}
\put(276,143){$v_1$}
\end{picture}
\caption{The process of fixed-length unfolding applied to a walk of length 27, producing a walk of length $27 + \lfloor 3\sqrt{27}\rfloor =42$. In the upper diagram the initial locations of $v_0$ and $v_1$ are indicated. In the lower diagram, the edges which are added at the unfolding points as well as at the end are dotted. Notice that, as this is a walk which ends adjacent to a surface vertex, in being unfolded it picks up two new surface contacts.}
\label{fig:unfolding}
\end{figure}

We make several observations here. Firstly, since we never have to perform the unfolding operation more than $\lfloor 3\sqrt{n}\rfloor$ times, the process takes a walk of length $n$ and maps it to an unfolded walk of length $n+\lfloor 3\sqrt{n}\rfloor+i$, where $i=0,1,2$ or $3$. Secondly, since the points at which we unfold (called $v_0$ and $v_1$ above) can never lie in the surface, the only time when we can add surface contacts is when we add steps at the end. This can only happen to a walk which ends adjacent to a surface vertex, and moreover, we only ever add two new contacts, since step 5 (b) involves adding steps strictly above the surface. Finally, the height of the endpoint of a walk is preserved by this operation. This is not strictly necessary for the results presented here, but it allows for other results previously shown for the hypercubic lattice in~\cite{Hammersley1982Selfavoiding,vanRensburg2006Selfavoiding} to be applied to the honeycomb lattice, so we include it for completeness.

One then ends up with something similar to the first inequality in~\eqref{eqn:ineq_unfolding}:

\[\mathbf C^+_n(y) \leq \left(1+\frac{1}{y^2}\right)\left(e^{c\sqrt{s_0(n)}}\mathbf U^+_{s_0(n)}(y) + e^{c\sqrt{s_1(n)}}\mathbf U^+_{s_1(n)}(y) + e^{c\sqrt{s_2(n)}} \mathbf U^+_{s_2(n)}(y) + e^{c\sqrt{s_3(n)}}\mathbf U^+_{s_3(n)}(y)\right),\]
where $s_i(n) := n+\lfloor3\sqrt{n}\rfloor + i$ and $c>0$ is a constant. It follows that
\begin{equation}\label{eqn:honeycomb_unfolding_ineq}
\mathbf C^+_n(y)\leq 4 \left(1+\frac{1}{y^2}\right)\max_{0\leq i\leq 3}\Big\{e^{c\sqrt{s_i(n)}} \mathbf U^+_{s_i(n)}(y)\Big\}. 
\end{equation}
To ease notation, define $m_n = s_{i_n}(n)$, where $0\leq i_n \leq 3$ and $i=i_n$ is the value which maximises $ e^{c\sqrt{s_i(n)}}\mathbf U^+_{s_i(n)}(y)$. Then~\eqref{eqn:honeycomb_unfolding_ineq} is
\begin{equation}\label{eqn:unfolding_simplified}
\mathbf C^+_n(y) \leq 4\left(1+\frac{1}{y^2}\right) e^{c\sqrt{m_n}}\mathbf U^+_{m_n}(y).
\end{equation}

Now the sequence $\{\mathbf U^+_{m_n}(y)\}_n$ is almost a subsequence of $\{\mathbf U^+_n(y)\}_n$ -- the terms come from the latter, in the right order, but may repeat up to three times. Similarly, the sequence $\{\mathbf U^+_{m_n}(y)^{1/m_n}\}_n$ is almost a subsequence of $\{\mathbf U^+_n(y)^{1/n}\}_n$. Since this repetition of elements makes no difference to matters of convergence, the convergence of $\{\mathbf U^+_n(y)^{1/n}\}_n$ to $\mu(y)$ ensures that $\{\mathbf U^+_{m_n}(y)^{1/m_n}\}_n$ converges to the same limit. Basic limit theorems and the fact that $n$ and $m_n$ differ by at most $\lfloor3\sqrt{n}\rfloor+3$ can then be used to show that $\{\mathbf U^+_{m_n}(y)^{1/n}\}_n$ also converges to $\mu(y)$.

We can then raise both sides of~\eqref{eqn:unfolding_simplified} to the power of $1/n$, and (noting that the other factors on the RHS will go to 1 in the limit) obtain
\[\limsup_{n\to\infty} \mathbf C^+_n(y)^{1/n} \leq \mu(y).\]
The other bound is far simpler: we obviously have $\mathbf U^+_n(y) \leq \mathbf C^+_n(y)$, and hence
\[\mu(y) \leq \liminf_{n\to\infty} \mathbf C^+_n(y)^{1/n}.\]
This completes the first part of the proposition.

The other results are elementary, and follow from a paper of Whittington~\cite{Whittington1975Selfavoiding}. In particular, the lower bound $\mu(y)\geq \sqrt{y}$ is obtained by considering walks which step along the surface.
\end{proof}

A quantity of much interest is the mean density of vertices in the surface, given by
\[\frac{1}{n}\frac{\sum_m mc^+_n(m)y^m}{\sum_m c_n^+(m)y^m} = \frac{y}{n}\frac{\partial\log \mathbf C^+_n(y)}{\partial y}.\]
In the limit of infinitely long walks, this density tends to\footnote{The exchange of the limit and the derivative is possible thanks to the convexity of $\log \mu(y)$, see for instance~\cite[Thm.~B7]{vanRensburg2000Statistical}.}
\[y\frac{\partial \log\mu(y)}{\partial y}.\]
From the behaviour of $\mu(y)$ given in Proposition~\ref{prop:FE_etc_existence}, it can be seen that the density of vertices in the surface is 0 for $y<\y$ and is positive for $y>\y$.

\subsection{Self-avoiding walks in a strip}\label{ssec:strip}

We now consider SAWs in a horizontal strip of the honeycomb lattice. In this geometry there are effectively two impenetrable surfaces with which walks can interact; we thus introduce a second surface fugacity $z$ associated with visits to vertices lying on the upper surface. As in the previous subsection, we take walks to start and end on mid-edges of the lattice, and set the starting point to be a horizontal mid-edge between two lower surface vertices. To make symmetry arguments more straightforward, we will remove the mid-edges protruding from the top and bottom of the strip.

We define an \emph{arch} to be a SAW which starts and ends on mid-edges at the bottom of the strip, and a \emph{bridge} to be a SAW which starts at the bottom and finishes at the top. Let $\hat{c}_{T,n}(l,m)$ be the number of $n$-step SAWs in a strip of height $T$ which visit $l$ vertices in the bottom surface and $m$ vertices in the top. Similarly, define $\hat{a}_{T,n}(l,m)$ and $\hat{b}_{T,n}(l,m)$ for arches and bridges respectively. (We use $\hat{a}$, $\hat{b}$ and $\hat{c}$ to distinguish these walks from those which end on protruding half-edges, which will be discussed in the next section.) See Figure~\ref{fig:strip_walks}. The partition function associated with SAWs in a strip is then
\[\hat{\mathbf C}_{T,n}(y,z) = \sum_{l,m}\hat{c}_{T,n}(l,m)y^l z^m,\]
and we similarly have $\hat {\mathbf A}_{T,n}(y,z)$ and $\hat{\mathbf B}_{T,n}(y,z)$ for arches and bridges.

\begin{figure}
\centering
\begin{picture}(400,120)
\put(0,0){\includegraphics[angle=90,width=400pt]{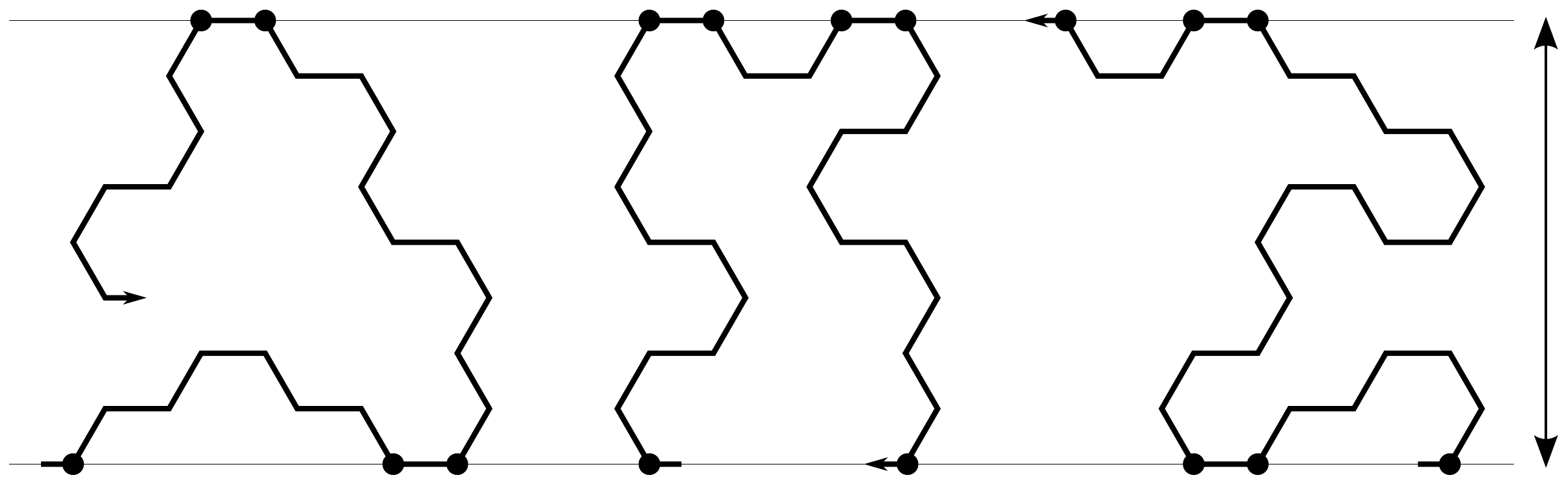}}
\put(400,58){$T$}
\end{picture}
\caption{Walks confined to a strip of width $T=9$ with weights attached to vertices along the top and bottom of the strip: a general walk, an arch, and a bridge.}
\label{fig:strip_walks}
\end{figure}

\begin{prop}\label{prop:strip_limits_existence}
For $y,z>0$, one has
\[\lim_{n\to\infty}\hat{\mathbf A}_{T,n}(y,z)^{1/n} = \lim_{n\to\infty}\hat{\mathbf B}_{T,n}(y,z)^{1/n} = \lim_{n\to\infty}\hat{\mathbf C}_{T,n}(y,z)^{1/n} := \mu_T(y,z),\]
where $\mu_T(y,z)$ is finite and non-decreasing in $y$ and $z$. By the symmetry of bridges,
\[\mu_T(y,z) = \mu_T(z,y),\]
and so in particular $\mu_T(y,1) = \mu_T(1,y)$. Finally, $\mu_T(1,y)$ is log-convex and thus is a continuous function of $\log y$.
\end{prop}
\begin{proof} As with Proposition~\ref{prop:FE_etc_existence}, the equivalent results for the $d$-dimensional hypercubic lattice have been previously proven~\cite{vanRensburg2006Selfavoiding}, also using unfolded walks and the process of unfolding. The same arguments we used for Proposition~\ref{prop:FE_etc_existence} apply here. The log-convexity result is easily adapted from~\cite[Thm.~6.3]{vanRensburg2006Selfavoiding}
\end{proof}
The utility of the result $\mu_T(y,1) = \mu_T(1,y)$ becomes immediately apparent when considering the geometry used in Subsection~\ref{ssec:general_y} (see also Figure~\ref{fig:domain_withweights}): if we only apply surface weights to one side of the strip, then it does not matter which side they go on. As discussed in Section~\ref{sec:identities}, it is convenient to place surface weights on the vertices of top boundary, rather than the bottom.

The next proposition concerns the behaviour of $\mu_T(1,y)$ as $T$ changes.
\begin{prop}\label{prop:muT_behaviour}
For $y>0$,
\[\mu_T(1,y) < \mu_{T+1}(1,y).\]
Moreover, as $T\to\infty$,
\[\mu_T(1,y) \to \mu(y),\]
where $\mu(y)$ is as defined in Proposition~\ref{prop:FE_etc_existence}.
\end{prop}
The proof is virtually identical to that of Proposition 7 in~\cite{Beaton2013Critical}, and we direct interested readers to that article. (The only difference is that the special ``prime'' arch used in that proof will necessarily be modified so as to fit on our lattice.)

The final result of this section concerns the properties of $\rho_T(y):=\mu_T(1,y)^{-1}$, which is the radius of convergence of the generating function
\[\hat{C}_T(x,y) := \sum_{n\geq0}\hat{\mathbf C}_{T,n}(1,y)x^n,\]
and of the similarly-defined functions $\hat{A}_T(x,y)$ and $\hat{B}_T(x,y)$.

\begin{cor}\label{cor:RCs_yT}
For $y>0$, the generating functions $\hat{A}_T(x,y)$, $\hat{B}_T(x,y)$ and $\hat{C}_T(x,y)$ all have the radius of convergence,
\[\rho_T(y) = \mu_T(1,y)^{-1}.\]
Moreover, $\rho_T(y)$ decreases to $\rho(y) := \mu(y)^{-1}$ as $T\to\infty$. In particular, $\rho_T(y)$ decreases to $\rho:=\mu^{-1}$ for $y\leq\y$. 

There exists a unique $y_T>0$ such that $\rho_T(y_T) = \x :=\mu^{-1}$. The series (in $y$) $\hat{A}_T(\x,y)$, $\hat{B}_T(\x,y)$ and $\hat{C}_T(\x,y)$ have radius of convergence $y_T$, and $y_T$ decreases to the critical fugacity $\y$ as $T\to\infty$.
\end{cor}
The proof is identical to that of Corollary 8 in~\cite{Beaton2013Critical}.

\section{Proof of the critical surface fugacity for SAWs}\label{sec:proof}


We now return to the identity~\eqref{eqn:ident_general_ny}, which we specialise to $n=0$ (note that $G_{T,L}(x,y;0) =1$):
\begin{multline}\label{eqn:identity_atn=0}
\cos\left(\frac{5\pi(2\pm1)}{16}\right)A^O_{T,L}(\x,y;0) + \cos\left(\frac{\pi(18\mp7)}{16}\right)A^I_{T,L}(\x,y;0)\\
+\cos\left(\frac{3\pi(2\mp1)}{16}\right)E_{T,L}(\x,y;0) + \frac{2}{\x}\cos\left(\frac{7(\pi\pm\theta)}{8}\right)P_{T,L}(\x,y;0)\\
+\left[\cos\left(\frac{\pi(2\pm1)}{16}\right)-\frac{(1-\x y-\x^2 y^2)\cos\left(\frac{\pi(14\mp1)}{16}\right)+\x^2 y^2\cos\left(\frac{5\pi(2\pm1)}{16}\right)}{\x y(1+\x y)}\right]B_{T,L}(\x,y;0)\\
=2\x \cos\left(\frac{\pi(2\pm1)}{16}\right).
\end{multline}
The identity of interest for SAWs is the second of this pair of equations. For brevity, we use the following shorthand:
\begin{align*}
c_A^O &:= 2\cos\left(\frac{5\pi}{16}\right) = \sqrt{2-\sqrt{2-\sqrt{2}}}, \qquad \qquad
c_A^I := 2\cos\left(\frac{7\pi}{16}\right) = \sqrt{2-\sqrt{2+\sqrt{2}}},\\
c_E &:= 2\cos\left(\frac{3\pi}{16}\right) = \sqrt{2+\sqrt{2-\sqrt{2}}}, \\
c_P &:= \frac{4}{\x}\cos\left(\frac{7\pi}{16}\right) = 2\sqrt{4+2\sqrt{2}-\sqrt{2\left(10+7\sqrt{2}\right)}},\\
c_G &:= 4\x\cos\left(\frac{\pi}{16}\right) = \sqrt{2 \left(4-2 \sqrt{2}+\sqrt{2 \left(2-\sqrt{2}\right)}\right)}, \text{ and}\\
c_B(y) &:= 2\cos\left(\frac{\pi}{16}\right)-\frac{2(1-\x y-\x^2y^2) \cos\left(\frac{15\pi}{16}\right)+2\x^2y^2\cos\left(\frac{5\pi}{16}\right)}{\x y(1+\x y)}\\
&= \frac{c_B}{\x y(1+\x y)} - \frac{\x y c_A^O}{1+\x y}, \qquad \text{ where } c_B := c_B(1) = 2\cos\left(\frac{\pi}{16}\right) = \sqrt{2+\sqrt{2+\sqrt{2}}}.
\end{align*}
For the rest of this section we will omit the $n=0$ argument from the generating functions. So~\eqref{eqn:identity_atn=0} can be written as
\begin{equation}\label{eqn:identity_shorthand}
c_A^OA_{T,L}^O(\x,y) + c_A^IA_{T,L}^I(\x,y) + c_EE_{T,L}(\x,y) + c_PP_{T,L}(\x,y) + c_B(y)B_{T,L}(\x,y) = c_G.
\end{equation}

We note here that $c_B(y)$ is a continuous and monotone decreasing function of $y$ for $y>0$, and that $c_B(y^\dagger)=0$ where
\[y^\dagger = \sqrt{\frac{2+\sqrt{2}}{1+\sqrt{2}-\sqrt{2+\sqrt{2}}}}.\]


For $0<y<y^\dagger$, every term in \eqref{eqn:identity_shorthand} is non-negative. Observe that $A^O_{T,L}$, $A^I_{T,L}$, $B_{T,L}$ and $P_{T,L}$ are increasing with $L$. (As $L$ increases these generating functions just count more and more objects.) We then see that for those values of $L$ satisfing $T+L\equiv1\,(\text{mod }2)$, $E_{T,L}$ must decrease as $L$ increases. It is thus valid to take the limit $L\to\infty$ of~\eqref{eqn:identity_shorthand} over the values of $L$ with $T+L\equiv1\,(\text{mod }2)$. But now $A^O_{T,L}$, $A^I_{T,L}$, $B_{T,L}$ and $P_{T,L}$ actually increase with $L$ regardless of whether $T+L\equiv1\,(\text{mod }2)$ or not, and so they have the same limits as $L\to\infty$ over any subsequence of $L$ values. Hence, we can in fact take the limit $L\to\infty$ of~\eqref{eqn:identity_shorthand} over all values of $L$. If we define
\[A^O_T(\x,y) := \lim_{L\to\infty} A^O_{T,L}(\x,y),\]
and similar limits for the other generating functions, then we obtain
\begin{equation}\label{eqn:ident_limit}
c_A^OA_{T}^O(\x,y) + c_A^IA_{T}^I(\x,y) + c_EE_{T}(\x,y) + c_PP_{T}(\x,y) + c_B(y)B_{T}(\x,y) = c_G.
\end{equation}
%
%

In this rest of this section, we will prove the following:
\begin{prop}\label{prop:critical_fugacity}
If it can be shown that 
\[B(\x,1):=\lim_{T\to\infty} B_T(\x,1) = 0\]
then $\y=y^\dagger$.
\end{prop}
\noindent The result that $B(\x,1)=0$ (Corollary~\ref{cor:Bto0}) is proved in the appendix, and in combination with Proposition~\ref{prop:critical_fugacity} completes the proof of Theorem~\ref{thm:main_thm}.

We begin by establishing a lower bound on $\y$.

\begin{lem}\label{lem:p:yc_geq_ydagger}
The critical surface fugacity $\y$ satisfies
\[\y\geq y^\dagger.\]
\end{lem}
\begin{proof}
Corollary~\ref{cor:RCs_yT} establishes the relationship between the generating functions $\hat{A}_T(\x,y)$, $\hat{B}_T(\x,y)$ and $\hat{C}_T(\x,y)$ and the critical fugacity $\y$. None of these generating functions feature in the identity~\eqref{eqn:ident_limit}. (Recall that $\hat{A}_T$ and $\hat{B}_T$ walks end on edges running along the bottom and top surfaces respectively, rather than on protruding mid-edges.) However, observe that there is a simple correspondence between $B_T$ walks and $\hat{B}_T$ walks: every $B_T$ walk can be obtained by reflecting the last step of a $\hat{B}_T$ walk, or by adding another step to the end of a $\hat{B}_T$ walk. Thus we have
\[B_T(x,y) = (1+xy)\hat{B}_T(x,y),\]
and so the generating functions $B_T(\x,y)$ and $\hat{B}_T(\x,y)$, viewed as series in $y$, have the same radius of convergence (namely $y_T$).

Now for $y<y^\dagger$ the identity \eqref{eqn:ident_limit} establishes the finiteness of $B_T(\x,y)$, and thus we see $y_T\geq y^\dagger$. By Corollary~\ref{cor:RCs_yT} it then follows that $\y\geq y^\dagger$.
%
%
\end{proof}

We now show that one of the generating functions in~\eqref{eqn:ident_limit} has disappeared in the limit $L\to\infty$.
\begin{cor}\label{cor:OO_Eto0}
For $0\leq y < y^\dagger$,
\[E_T(\x,y) := \lim_{L\to\infty} E_{T,L}(\x,y) =0,\]
and hence
\begin{equation}\label{eqn:ident_withoutE}
c_A^O A_T^O(\x,y) + c_A^I A_T^I(\x,y) + c_P P_T(\x,y) + c_B(y) B_T(\x,y) = c_G.
\end{equation}
\end{cor}
\begin{proof}
By Corollary~\ref{cor:RCs_yT}, $y_T$ is the radius of convergence of $\hat{C}_T(\x,y)$. Since $y_T\geq \y\geq y^\dagger$ (Lemma~\ref{lem:p:yc_geq_ydagger}), it follows that $\hat{C}_T(\x,y)$ is convergent for $0\leq y<y^\dagger$. Now
\[\sum_{L}E_{T,L}(\x,y) \leq \hat{C}_T(\x,y)<\infty,\]
as each walk counted by $E_{T,L}$, for every value of $L$, will also be counted by $\hat{C}_T$. The corollary follows immediately.
\end{proof}

We note here that $A_T^O(\x,y) \leq \x \hat{C}_T(\x,y)$ (since any walk counted by $A_T^O$ can be obtained by attaching a step to a unique walk counted by $\hat{C}_T$), and likewise for $A_T^O$ and $P_T$. Hence all the generating functions featured in~\eqref{eqn:ident_withoutE} have radius of convergence at least $y_T$.


Now consider the $y=1$ case of~\eqref{eqn:ident_withoutE}:
\[c_A^O A_T^O(\x,1) + c_A^I A_T^I(\x,1) + c_P P_T(\x,1) + c_B B_T(\x,1) = c_G.\]
Since $A_T^O(\x,1)$, $A_T^I(\x,1)$ and $P_T(\x,1)$ all increase with $T$ (as $T$ increases these generating functions count more and more objects), and since they are all bounded by this identity, it follows that they all have limits as $T\to\infty$. Then $B_T(\x,1)$ must decrease as $T$ increases, and it too has a limit as $T\to\infty$. As indicated in Proposition~\ref{prop:critical_fugacity}, we denote this limit
\[B(\x,1) := \lim_{T\to\infty}B_T(\x,1).\]

\begin{proof}[Proof of Proposition~\ref{prop:critical_fugacity}] Assume now that $B(\x,1)=0$. 
Any walk counted by $A_{T+1}^O(\x,y)$ which has contacts with the top boundary can be factored into two pieces by cutting it at the mid-edge immediately following its last surface contact. (See Figure~\ref{fig:arch_decomp}.) The first piece, after reflecting the last step, is an object counted by $B_{T+1}(\x,y)$, while the second piece (with its direction reversed) will be counted by $(1+\x)B_T(\x,1)/2$. Thus we obtain
\begin{align*}
A_{T+1}^O(\x,y) - A_T^O(\x,1) &\leq \frac{1+\x}{2}\cdot B_{T+1}(\x,y)B_T(\x,1)\\
&\leq B_{T+1}(\x,y)B_T(\x,1)
\end{align*}
This inequality is valid in the domain of convergence of the series it involves, that is, for $y<y_{T+1}$. Using similar arguments we can obtain the equivalent inequality for $A_{T+1}^I(\x,y)$ and $P_{T+1}(\x,y)$.

\begin{figure}
\centering
\begin{picture}(360,120)
\put(0,0){\includegraphics[angle=90,width=360pt]{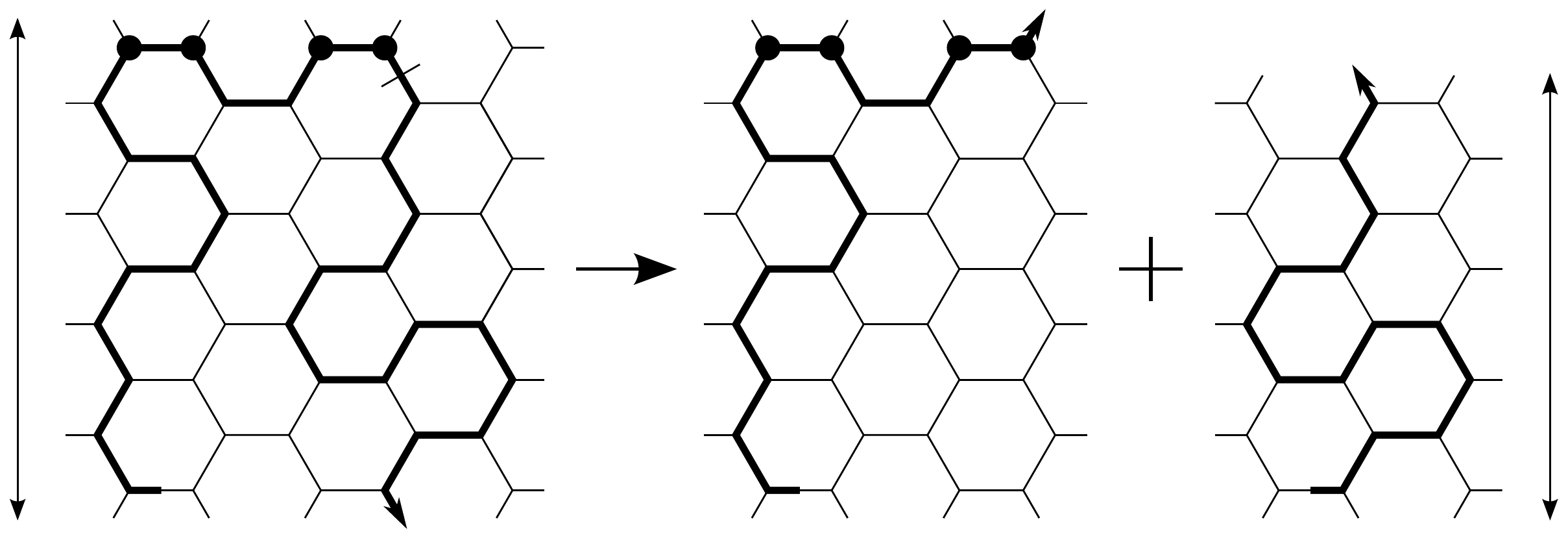}}
\put(-28,58){$T+1$}
\put(365,52){$T$}
\end{picture}
\caption{Factorisation of a walk counted by $A^O_{T+1}$ into two bridges.}
\label{fig:arch_decomp}
\end{figure}

Combining this decomposition for $A^O_{T+1},A^I_{T+1}$ and $P_{T+1}$, we find for $0\leq y < y_{T+1}$,
\begin{multline}\label{eqn:p:OO_decomp_combined}
c_A^O[A_{T+1}^O(\x,y) - A_T^O(\x,1)] + c_A^I[A_{T+1}^I(\x,y) - A_T^I(\x,1)] + c_P[P_{T+1}(\x,y) - P_T(\x,1)] \\
\leq (c_A^O + c_A^I + c_P) B_{T+1}(\x,y)B_T(\x,1).
\end{multline}
Using~\eqref{eqn:ident_withoutE} to eliminate the $A^O,A^I$ and $P$ terms, we obtain
\begin{equation*}
c_BB_T(\x,1) - c_B(y)B_{T+1}(\x,y) \leq (c_A^O + c_A^I + c_P) B_{T+1}(\x,y)B_T(\x,1),
\end{equation*}
and hence
\begin{equation}\label{eqn:p:OO_decomp_rearranged}
0 \leq \frac{1}{B_{T+1}(\x,y)} \leq \frac{(c_A^O + c_A^I + c_P)}{c_B} + \frac{c_B(y)}{c_BB_T(\x,1)}.
\end{equation}
In particular, for $0\leq y <\y = \lim_{T\to\infty} y_T$ and for any $T$,
\begin{equation}\label{eqn:p:OO_final_usefulbound}
0 \leq \frac{\x(c_A^O + c_A^I + c_P)}{c_B} + \frac{c_B(y)}{c_BB_T(\x,1)}.
\end{equation}
Now suppose that $\y > y^\dagger$, and consider what happens as we take $T\to\infty$. By assumption, $B_T(\x,1)\to 0$. For any $y^\dagger<y<\y$, the RHS of~\eqref{eqn:p:OO_final_usefulbound} must go to $-\infty$, because $c_B(y)<0$ for $y>y^\dagger$. This is clearly a contradiction, and we are forced to conclude $\y\leq y^\dagger$, and hence $\y=y^\dagger$.
%
\end{proof}

\section*{Acknowledgments}

I thank Murray Batchelor for suggesting this problem, and Tony Guttmann for helpful conversations. I received support from the ARC Centre of Excellence for Mathematics and Statistics of Complex Systems (MASCOS), as well as the Australian Mathematical Society (AustMS) in the form of a Lift-Off Fellowship. Part of this work was carried out while I was a guest of the Mathematical Sciences Research Institute (MSRI) in Berkeley, CA, during the Spring 2012 Random Spatial Processes Program, and I thank the Institute for its hospitality and the NSF (grant DMS-0932078) for its financial support.

\section*{Appendix}

We begin this appendix with some new definitions and notation, before stating its main theorem. Its structure is almost identical to the appendix of~\cite{Beaton2013Critical}, and thus we will omit a number of details which can be found in that article, and instead focus mainly on the minor changes which need to be made.

As usual, we orient the honeycomb lattice so that it contains horizontal edges, and scale it so that edges have unit length. The set of mid-edges of the lattice is denoted by $\mathbb{H}$. The edges of the lattice are oriented in three different directions; recall that we refer to a south-west/north-east (resp.~north-west/south-east) oriented edge as \emph{positive} (resp.~\emph{negative}), and likewise a positive or negative mid-edge is the mid-edge of a positive or negative edge. The lattice has an origin in $\mathbb{H}$, which we will fix to lie on a positive mid-edge $a_0$. We denote by $(\mathbf x(v),\mathbf y(v))$ the coordinates of a point $v\in\mathbb{C}$. We consider self-avoiding walks to start and end at mid-edges. A self-avoiding walk $\gamma$ of length $n$ is denoted by the sequence $(\gamma_0,\ldots,\gamma_n)$ of its mid-edges. As usual we denote by $|\gamma|$ the length of a walk. To lighten notation, we often omit floor symbols, especially in indices: for instance, $\gamma_t$ should be understood as $\gamma_{\lfloor t\rfloor}$. When referencing generating functions used in earlier sections, we will omit $y$ and $n$ arguments, which will always be taken to be 1 and 0 respectively.

We have so far referred to bridges in several contexts (specifically, we have referred both to the objects counted by $B_T$ and $\hat{B}_T$ as bridges of height $T$). In this appendix we will work with a new class of SAWs: we define a $\emph{PP-bridge}$ $\gamma$ to be a SAW which starts and ends on positive mid-edges and satisfies $\mathbf y(\gamma_0) < \mathbf y(\gamma_i) < \mathbf y(\gamma_n)$ for $0<i<n$. (The PP stands for \emph{positive-positive} -- we will later introduce some other types of walks using a similar naming convention.) The set of PP-bridges of length $n$ is denoted by $\SAPPn$. The \emph{height} $\height(\gamma)$ of a PP-bridge $\gamma$ is the length of the shortest PP-bridge $\gamma'$ satisfying $\mathbf y(\gamma_0) = \mathbf y(\gamma'_0)$ and $\mathbf y(\gamma_n) = \mathbf y(\gamma'_n)$.

The central result of this appendix is the following theorem.
\begin{thm}\label{thm:appendix_thm}
Let $\PPT(x)$ to be the generating function of PP-bridges of height $T$, that is,
\[\PPT(x) := \sum_{n\geq0}\sum_{\substack{\gamma\in\SAPPn\\H(\gamma)=T}} x^n.\]
Then
\[\lim_{T\to\infty} \PPT(\x) = 0,\]
where $\x=1/\sqrt{2+\sqrt{2}}$.
\end{thm}
Before proceeding with the proof, we present a corollary which relates this result to those of the previous sections.
\begin{cor}\label{cor:Bto0}
\[B(\x) := \lim_{T\to\infty}B_T(\x) = 0.\]
\end{cor}
\begin{proof}
We split the walks counted by $B_T(\x)$ in two ways (refer to Figure~\ref{fig:domain_noweights}). Let $\overleftarrow{B}^+_T(\x)$ count those walks which pass through $a^-$ and finish at a mid-edge in $\beta^+$, and similarly $\overleftarrow{B}^-_T(\x)$ counts those which pass through $a^-$ and finish at a mid-edge in $\beta^-$. Alternatively, $\overrightarrow{B}^+_T(\x)$ and $\overrightarrow{B}^-_T(\x)$ count those walks which pass through $a^+$ and finish at a mid-edge in $\beta^+$ or $\beta^-$ respectively. Of course, we have 
\begin{align*}
B_T(\x) &= \overleftarrow{B}^+_T(\x) + \overleftarrow{B}^-_T(\x)+\overrightarrow{B}^+_T(\x) + \overrightarrow{B}^-_T(\x),\\
\intertext{and then by reflective symmetry,}
&= 2\left(\overleftarrow{B}^+_T(\x) + \overrightarrow{B}^+_T(\x)\right).
\end{align*}
Now $\PPT(\x) = \overleftarrow{B}^+_T(\x) + \x \overrightarrow{B}^+_T(\x)$ (if an $\PPT$ walk starts with a left turn, it is a $\overleftarrow{B}^+_T$ walk with the first step reflected; if it starts with a right turn, it is a $\overrightarrow{B}^+_T$ walk with an extra step attached to the start), and so
\begin{equation}\label{eqn:BTRT_first}
2\PPT(\x) = 2\left(\overleftarrow{B}^+_T(\x) + \x \overrightarrow{B}^+_T(\x)\right) \leq 2\left(\overleftarrow{B}^+_T(\x) +  \overrightarrow{B}^+_T(\x)\right) = B_T(\x).
\end{equation}
On the other hand,
\begin{equation}\label{eqn:BTRT_second}
\frac{\x}{2}B_T(\x) = \x\left(\overleftarrow{B}^+_T(\x) +  \overrightarrow{B}^+_T(\x)\right) \leq \overleftarrow{B}^+_T(\x) + \x \overrightarrow{B}^+_T(\x) = \PPT(\x).
\end{equation}
Combining~\eqref{eqn:BTRT_first} and~\eqref{eqn:BTRT_second}, we have
\[2\PPT(\x) \leq B_T(\x) \leq \frac{2}{\x}\PPT(\x).\]
Applying Theorem~\ref{thm:appendix_thm} then shows that $B_T(\x)\to 0$ as $T\to\infty$.
\end{proof}
\noindent Combining Corollary~\ref{cor:Bto0} with Proposition~\ref{prop:critical_fugacity} will complete the proof of Theorem~\ref{thm:main_thm}.

We present now some further definitions. The set $\mathbf R_\gamma$ of \emph{renewal points} of $\gamma\in\SAPPn$ is $\{\gamma_0,\gamma_n\}$, together with the set of points of the form $\gamma_i$ with $0< i< n$, for which $\gamma_{[0,i]} := (\gamma_0,\ldots,\gamma_i)$ and $\gamma_{[i,n]} := (\gamma_i,\ldots,\gamma_n)$ are PP-bridges. We denote by $\mathbf{r}_0(\gamma)$, $\mathbf{r}_1(\gamma),\ldots$ the indices of the renewal points. That is, $\mathbf{r}_0(\gamma) = 0$ and $\mathbf{r}_{k+1}(\gamma) = \inf\{j>\mathbf{r}_k(\gamma):\gamma_j\in \mathbf R_\gamma\}$. When no confusion is possible, we often denote $\mathbf{r}_k(\gamma)$ by just $\mathbf{r}_k$.

An PP-bridge $\gamma\in\SAPPn$ is \emph{irreducible} if its only renewal points are $\gamma_0$ and $\gamma_n$. Let iSAPP be the set of irreducible PP-bridges of arbitrary length starting at $a$. Every PP-bridge $\gamma$ is the concatenation of a finite number of irreducible PP-bridges, the decomposition is unique and the set $\mathbf{R}_\gamma$ is the union of the initial and terminal points of the PP-bridges that comprise this decomposition.

Kesten's relation for irreducible bridges~\cite{Kesten1963Number} on the hypercubic lattice can be adapted to our lattice without difficulty. It gives
\[\sum_{\gamma\in\iSAPP}\x^{|\gamma|} = 1.\]
This enables us to define a probability measure $\PiSAPP$ on iSAPP by setting $\PiSAPP(\gamma) = \x^{|\gamma|}$. Let $\PNiSAPP$ denote the law on semi-infinite walks $\gamma:\mathbb{N}\to\mathbb{H}$ formed by the concatenation of infinitely many samples $\gamma^{[1]},\gamma^{[2]},\ldots$ of $\PiSAPP$. We refer to~\cite[Section 8.3]{Madras1993Selfavoiding} for details of related measures in the case of $\mathbb{Z}^d$. The definition of $\mathbf{R}_\gamma$ and the indexation of renewal points extend to this context (we obtain an infinite sequence $(\mathbf{r}_k)_{k\in\mathbb{N}}$).

Observe that a PP-bridge $\gamma$ of length $n$ has height $\height(\gamma) = \frac{2}{\sqrt{3}}\mathbf y(\gamma_n)$. We define the \emph{width} of $\gamma$ to be
\[\width(\gamma) = \frac{2}{3}\max\{\mathbf x(\gamma_k)-\mathbf x(\gamma_{k'}),0\leq k,k'\leq n\}.\]
Intuitively, the width of a PP-bridge is the total number of columns of cells it spans.

The next result, equivalent to Lemma 11 of~\cite{Beaton2013Critical}, relates the limiting value of $\PPT(\x)$ to the average height of irreducible PP-bridges.
\begin{lem}\label{lem:RT_expectedheight}
As $T\to\infty$,
\[\PPT(\x) \to \frac{1}{\mathbb{E}_{\rm iSAPP}(\height(\gamma))}.\]
\end{lem}
\begin{proof}
The result follows from standard renewal theory. We can for instance apply~\cite[Theorem 4.2.2(b)]{Madras1993Selfavoiding} to the sequence
\[f_T := \sum_{\substack{\gamma\in\iSAPP\\H(\gamma)=T}} \x^{|\gamma|}.\]
Indeed, with the notation of this theorem, $v_T= \PPT(\x)$ and $\sum_k kf_k = \EiSAPP(\height(\gamma))$.
\end{proof}
Thus Theorem~\ref{thm:appendix_thm} is equivalent to
\[\EiSAPP(H(\gamma)) = \infty.\]
We prove this by contradiction, in the same way as Theorem 10 of~\cite{Beaton2013Critical}. Assuming $\EiSAPP(\height(\gamma))$ is finite, we first show that $\EiSAPP(\width(\gamma))$ is also finite. Then, we show that under these two conditions, an infinite PP-bridge is very narrow. The last step of the proof involves demonstrating that this leads to a contradiction. The argument uses a ``stickbreak'' operation, which perturbs a PP-bridge by selecting a subpath and rotating it clockwise by $\frac{\pi}{3}$. The new path is a self-avoiding PP-bridge for an adequately chosen subpath, but its width is relatively large, contradicting the fact that PP-bridges are narrow.

The proof of Theorem 10 of~\cite{Beaton2013Critical} was greatly inspired by a recent paper of Duminil-Copin and Hammond~\cite{DuminilCopin2012Selfavoiding}, where self-avoiding walks are proved to be sub-ballistic. We will refer the reader to~\cite{Beaton2013Critical} where appropriate, as many components of the proof in that article require no modification in order to apply here. For the next result (equivalent to Proposition 12 and Lemmas 13 and 14 of~\cite{Beaton2013Critical}), however, several minor changes are required, and so we present the proof in full.

\begin{prop}\label{prop:H_W_finite}
If $\EiSAPP(\height(\gamma))<\infty$, then $\EiSAPP(\width(\gamma))<\infty$.
\end{prop}
\begin{proof}
We return to the special domain $D_{T,L}$ as defined in Section~\ref{sec:identities} (see Figure~\ref{fig:domain_noweights}), without any surface weights. Recall the identity~\eqref{eqn:unweighted_identity} for this domain:
\[c_A^OA_{T,L}^O(\x) + c_A^IA_{T,L}^I(\x) + c_EE_{T,L}(\x) + c_BB_{T,L}(\x) + c_PP_{T,L}(\x) = c_G.\]
As in Section~\ref{sec:proof}, we would like $E_{T,L}(\x)$ to tend to 0 as the size of the domain increases. We previously showed this is the case as $L$ increases for fixed $T$, but now we wish to let both $T$ and $L$ increase. Recall that we defined an \emph{arch} to be a SAW in a strip which starts and ends at horizontal mid-edges on the bottom of the strip. Such a definition obviously generalises to walks in the upper half-plane. For even $L\in\mathbb{N}$, let $\mathbf{a}_L(x)$ be the generating function of arches in the upper half-plane which end $L$ columns to the right of their starting point. 

We note here that the generating function of all arches,
\[\hat{A}(x) = \sum_{L\geq0}\mathbf{a}_L(x),\]
satisfies
\[\hat{A}(\x) \leq \frac{1}{1+\x}\left(A^O(\x) + A^I(\x)\right)+1,\]
where $A^O(\x) = \lim_{T\to\infty}A^O_T(\x)$ and $A^I(\x) = \lim_{T\to\infty} A^I_T(\x)$. (To see this, observe that we can reflect the last half-edge of an arch to produce a walk counted by $A^O$ (or $A^I$, depending on the direction of the last step), or we could add a step to produce a walk counted by $A^I$ (or $A^O$). The empty arch is an exception so it is treated separately.) As discussed in Section~\ref{sec:proof}, $A^O(\x)$ and $A^I(\x)$ are finite, and so we see $\hat{A}(\x)<\infty$.

For $m\in\mathbb{N}$, let $\mathbf{e}^+_m(x)$ be the generating function of walks in $D_{T,L}$ which start at $a$ and end on the $m^{\text{th}}$ row of $\epsilon^+$, so that $E_{T,L}(\x) = 2\sum_{m\leq\lfloor\frac{T}{2}\rfloor}\mathbf{e}_m^+(\x)$. Using a reflection argument and the Cauchy-Schwarz inequality, we find
\begin{equation}\label{eqn:boundE_witharches}
\left(E_{T,L}(\x)\right)^2 \leq 4\Big\lfloor\frac{T}{2}\Big\rfloor\sum_{m\leq\lfloor\frac{T}{2}\rfloor}\left(\mathbf{e}_m^+(\x)\right)^2 \leq 4\Big\lfloor\frac{T}{2}\Big\rfloor\mathbf{a}_{2L+2}(\x).
\end{equation}
(The second inequality comes from the fact that we can concatenate two walks counted by $\mathbf{e}_m^+$ (after reflecting the second one) to produce an arch.)

Assume that we couple $T\equiv T_k$ and $L\equiv L_k$ so that both tend to infinity as $k$ grows, and $T\mathbf{a}_{2L+2}(\x)\to0$. Then $E_{T,L}(\x)$ tends to 0. Moreover, $A^O_{T,L}(\x)$, $A^I_{T,L}(\x)$ and $P_{T,L}(\x)$ increase with $T$ and $L$, and converge respectively to $A^O(\x)$, $A^I(\x)$ and $P(\x)$. Then $B_{T,L}(\x)$ also converges, and its limit must be
\begin{align}
\lim_{k\to\infty}B_{T_k,L_k}(\x)  = B(\x) &= \lim_{T\to\infty}B_T(\x)\notag\\
&\geq \lim_{T\to\infty} 2\PPT(\x)\notag\\
&>0,\label{eqn:BTk>0}
\end{align}
where the last two inequalities follow from~\eqref{eqn:BTRT_first} and by assumption, respectively.

We now return to random infinite bridges and use them to give an upper bound on $B_{T,L}(\x)$. We consider again the domain $D_{T,L}$, and denote by $a^*$ the external mid-edge adjacent to the vertex $a^-$ (not shown in Figures~\ref{fig:domain_noweights} and~\ref{fig:domain_withweights}). We then define $\PPTL(x)$, in the obvious way, to be the generating function of PP-bridges of height $T$ in $D_{T,L}$ which begin at $a^*$ and end at the top of the rectangle. By the same arguments used to obtain~\eqref{eqn:BTRT_second},
\begin{equation}\label{eqn:bound_BTL_RTL}
B_{T,L}(\x) \leq \frac{2}{\x}\PPTL(\x).
\end{equation}
Let $0<\delta<1/\EiSAPP(\mathsf H(\gamma))$. We have
\begin{align*}
\PPTL(\x) &= \sum_{\gamma:a^*\to\beta^+\bigcup\beta^-} \x^{|\gamma|}\\
&\leq \PNiSAPP(\exists n\in\mathbb{N}:\height(\gamma_{[0,\mathbf{r}_n]}) = T \text{ and } \width(\gamma_{[0,\mathbf{r}_n]}\leq 2L+1)\\
&\leq \PNiSAPP(\height(\gamma_{[0,\bfr_{\delta T}]})\geq T) + \PNiSAPP(\exists n\geq\delta T:\height(\gamma_{[0,\bfr_n]})=T \text{ and }\width(\gamma_{[0,\bfr_n]})\leq 2L+1).
\end{align*}
Let $\gamma^{[i]}$ be the $i^{\text{th}}$ irreducible PP-bridge of $\gamma$. Since the $\gamma^{[i]}$s are independent, we obtain
\begin{align*}
\PPTL(\x) &\leq \PNiSAPP(\height(\gamma_{[0,\bfr_{\delta T}]})\geq T) + \PNiSAPP(\forall i\leq \delta T, \width(\gamma^{[i]})\leq 2L+1)\\
&= \PNiSAPP(\height(\gamma_{[0,\bfr_{\delta T}]})\geq T) + \PiSAPP(\width(\gamma)\leq 2L+1)^{\delta T}\\
&\leq \PNiSAPP(\height(\gamma_{[0,\bfr_{\delta T}]})\geq T) + \exp(-\delta T\PiSAPP(\width(\gamma)>2L+1)).
\end{align*}

Note that
\[\height(\gamma_{[0,\bfr_{\delta T}]}) = \sum_{i=1}^{\delta T}\height(\gamma^{[i]}).\]
Hence we can use the law of large numbers, together with the fact that $\delta\EiSAPP(\height(\gamma))<1$, to see that $\PNiSAPP(\height(\gamma_{[0,\bfr_{\delta T}]})\geq T)$ tends to 0 as $T\to\infty$. So if we can couple $T\equiv T_k$ and $L\equiv L_k$ in such a way that $T\PiSAPP(\width(\gamma)>2L+1)$ tends to infinity, then $\PPTL(\x)$ tends to 0, and then $B_{T,L}(\x)$ tends to 0 by~\eqref{eqn:bound_BTL_RTL}.

We now argue \emph{ad absurdum}. Assume that $\EiSAPP(\width(\gamma))=\infty$. Then
\[\limsup_{L\to\infty}\frac{\PiSAPP(\width(\gamma)>2L+1)}{\mathbf{a}_{2L+2}(\x)}=\infty,\]
since $\mathbf{a}_{2L}(\x)$ is the term of a converging series (namely the generating function $\hat{A}(\x)$ of arches) and $\PiSAPP(\width(\gamma)>L)$ is non-increasing in $L$ and is the term of a diverging series (in particular, it sums to $\EiSAPP(\width(\gamma))=\infty$.) Let $(L_k)_k$ be a sequence such that
\[\lim_{k\to\infty}\frac{\PiSAPP(\width(\gamma)>2L_k+1)}{\mathbf{a}_{2L_k+2}(\x)}=\infty,\]
and take
\[T_k = \Bigg\lfloor\frac{1}{\sqrt{\mathbf{a}_{2L_k+2}(\x)\PiSAPP(\width(\gamma)>2L_k+1)}}\Bigg\rfloor.\]
Then 
\[T_k\PiSAPP(\width(\gamma)>2L_k+1) \to \infty \quad \text{and}\quad T_k\mathbf{a}_{2L_k+2}(\x)\to 0.\]
Now by~\eqref{eqn:BTk>0} we have $\lim_{k\to\infty}B_{T_k,L_k}(\x)>0$, but we also have
\[\lim_{k\to\infty} B_{T_k,L_k}(\x) \leq \lim_{k\to\infty}\frac{2}{\x}P\!P_{T_k,L_k}(\x) = 0.\]
We thus have a contradiction, and conclude that $\EiSAPP(\width(\gamma))<\infty$.
\end{proof}

Let $\Omega$ be the set of bi-infinite walks $\gamma:\mathbb{Z}\to\mathbb{H}$ such that $\gamma_0=a_0$. Let $(\gamma^{[i]},i\in\mathbb{Z})$ be a bi-infinite sequence of irreducible PP-bridges sampled independently according to $\PiSAPP$. Let $\PZiSAPP$ denote the law on $\Omega$ formed by concatenating the PP-bridges $\gamma^{[i]}$ in such a way that $\gamma^{[1]}$ starts at $a_0$. Let $\mathcal{F}$ be the $\sigma$-algebra generated by events depending on a finite number of vertices of the walk.

We extend the indexation of renewal points to these bi-infinite PP-bridges (we obtain a bi-infinite sequence $(\bfr_n(\gamma))_{n\in\mathbb{Z}}$ such that $\bfr_0(\gamma)=0$). Let $\tau:\Omega\to\Omega$ be the \emph{shift} defined by $\tau(\gamma)_i = \gamma_{i+\bfr_1(\gamma)} - \gamma_{\bfr_1(\gamma)}$ for every $i\in\mathbb{Z}$. (This is only defined if $\bfr_1$ exists, but this is the case with probability 1 under $\PZiSAPP$.) The shift translates the walk so that $\bfr_1(\gamma)$ is now at the origin $a_0$ of the lattice. Note that $\bfr_i(\tau(\gamma)) = \bfr_{i+1}(\gamma)-\bfr_1(\gamma)$. Let $\sigma$ denote the \emph{rotation} through angle $\pi$ about the origin $a_0$.

The following proposition is equivalent to Proposition 15 of~\cite{Beaton2013Critical}. The only difference is the fact that here $\sigma$ is a \emph{rotation}, whereas in~\cite{Beaton2013Critical} it is a \emph{reflection}. This does not affect the proof at all, and thus we direct the reader to that article for further details.

\begin{prop}\label{prop:ergodic_etc}
The measure $\PZiSAPP$ satisfies the following properties.
\begin{itemize}
\item[$({\rm P}_1)$] It is invariant under the shift $\tau$.
\item[$({\rm P}_2)$] The shift $\tau$ is ergodic for $(\Omega,\mathcal{F},\PZiSAPP)$.
\item[$({\rm P}_3)$] Under $\PZiSAPP$, the random variables $(\sigma\gamma_n)_{n\leq0}$ and $(\gamma_n)_{n\leq0}$ are independent and identically distributed.
\end{itemize}
\end{prop}

Renewal points separate a walk into two pieces, located above and below the point. We now introduce a more restrictive notion, illustrated in Figure~\ref{fig:stickbreak}. A mid-edge $\gamma_k$ of a walk $\gamma$ is said to be a \emph{diamond point} if
\begin{itemize}
\item it is a renewal point of $\gamma$, and
\item the walk is contained in the cone
\[\left(\left(\gamma_k-\textstyle{\frac{1}{4}}-\frac{\sqrt{3}}{4}\ii\right) + \mathbb{R}_+\e^{\ii\pi/3} + \mathbb{R}_+ \e^{2\ii\pi/3}\right) \bigcup \left(\left(\textstyle\gamma_k+\frac{1}{4}+\frac{\sqrt{3}}{4}\ii\right) - \mathbb{R}_+ \e^{\ii\pi/3} - \mathbb{R}_+\e^{2\ii\pi/3}\right).\]
\end{itemize}
We denote the set of diamond points of $\gamma$ by $\mathbf{D}_\gamma$. 

\begin{figure}
\centering
\includegraphics[angle=90,width=320pt]{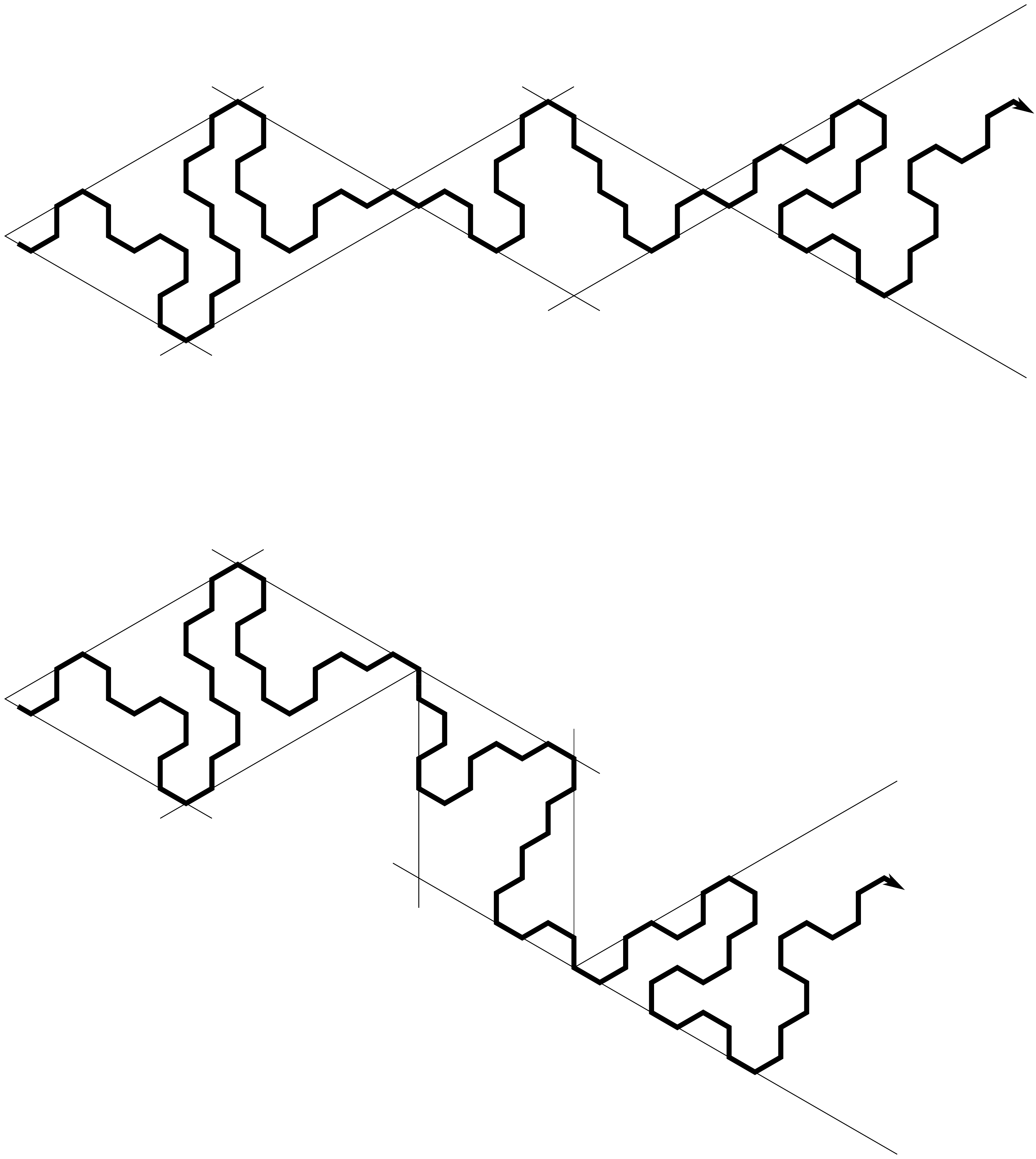}
\caption{An PP-bridge having three diamond points, and the same bridge after application of the $\mathsf{StickBreak}$ operation.}
\label{fig:stickbreak}
\end{figure}

The following proposition, equivalent to Proposition 16 of~\cite{Beaton2013Critical}, tells us that under our assumption $\EiSAPP(\height(\gamma))<\infty$, a positive fraction of renewal points are diamond points. As usual, the proof is very similar to~\cite{Beaton2013Critical}, but there are a sufficient number of differences that we will present the whole proof.

\begin{prop}\label{prop:diamond_positive_density}
If $\EiSAPP(\height(\gamma))<\infty$, then there exists $\delta>0$ such that
\[\PNiSAPP\left(\liminf_{n\to\infty}\frac{|{\mathbf D}_\gamma\cap\{0,\ldots,\bfr_n\}|}{n}\geq\delta\right)=1.\]
\end{prop}

Before proving this proposition we present a useful lemma:
\begin{lem}\label{lem:expected_x_0}
If $\EiSAPP(\height(\gamma))<\infty$, then $\EiSAPP(\mathbf x(\gamma_{|\gamma|}))=0$.
\end{lem}
\begin{proof}
This result is trivial for objects whose law is invariant under reflection through the imaginary axis, but unfortunately PP-bridges do not satisfy this criterion. Instead, we introduce a further decomposition of PP-bridges. In this proof we will be considering walks which start on a negative mid-edge, but hope that the interpretations of $\mathbf x$ and $\mathbf y$ coordinates remain clear. (The starting point of a walk will usually be assumed to be the origin.)

We define a \emph{P-bridge} (resp.~\emph{N-bridge}) $\gamma$ to be a SAW of length $n$ which starts on a positive (resp.~negative) mid-edge and ends on a positive or negative mid-edge, and satisfies $\mathbf y(\gamma_0) < \mathbf y(\gamma_i) < \mathbf y(\gamma_n)$ for $0<i<n$. An \emph{x-renewal point} of a P-bridge (resp.~N-bridge) $\gamma$ is one of $\{\gamma_0,\gamma_n\}$, or a point $\gamma_i\in\gamma$ such that $\gamma_{[0,i]}$ is a P-bridge (resp.~N-bridge) and $\gamma_{[i,n]}$ is either is a P- or N-bridge. Clearly PP-bridges are a subset of P-bridges, and the renewal points of a PP-bridge are a subset of its x-renewal points. A P-bridge or N-bridge $\gamma$ is \emph{x-irreducible} if its only x-renewal points are $\gamma_0$ and $\gamma_n$. (Note that for PP-bridges, x-irreducibility is a stronger condition than reducibility.)

We then define a \emph{PN-bridge} to be a P-bridge which ends on a negative mid-edge. Likewise, an \emph{NP-bridge} is an N-bridge which ends on a positive mid-edge, and an NN-bridge is an N-bridge which ends on a negative mid-edge. Let xSAP denote the set of x-irreducible P-bridges, and similarly define xSAN, xSAPP, xSANN, xSAPN and xSANP. In Figure~\ref{fig:examples_PP_etc} we illustrate examples of some of these objects.

\begin{figure}
\centering
\begin{picture}(400,120)
\put(0,20){\includegraphics[width=400pt]{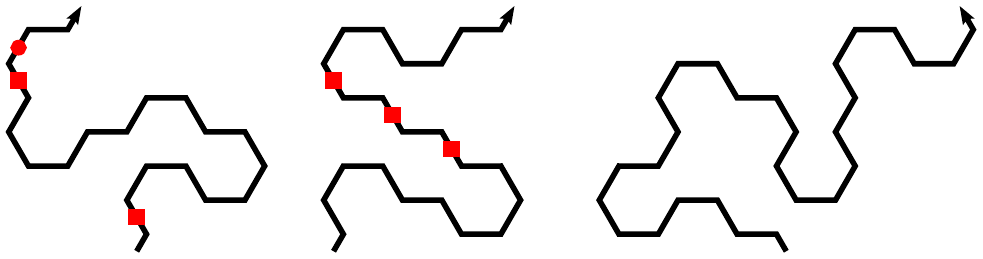}}
\put(50,0){(a)}
\put(170,0){(b)}
\put(310,0){(c)}
\end{picture}
\caption{(a) and (b): PP-bridges, and (c): an NN-bridge. For the PP-bridges, circles indicate renewal points, and squares indicate x-renewal points which are not also renewal points. Note that (b) is irreducible but not x-irreducible, while (c) is x-irreducible.}
\label{fig:examples_PP_etc}
\end{figure}
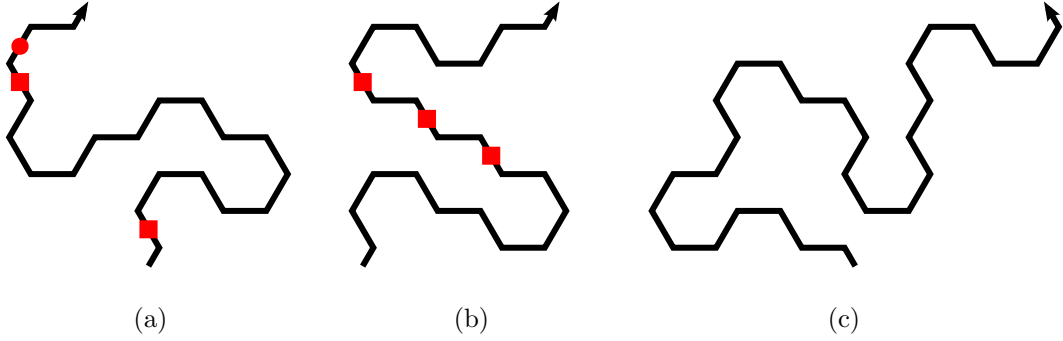

We denote by $\mathbb{P}^{\rm PN}_n$ the uniform probability measure on PN-bridges of length $n$. First, observe that 
\begin{equation}\label{eqn:Ex_NP_PN}
\mathbb{E}^{\rm PN}_n(\mathbf x(\gamma_{n})) =0.
\end{equation}
To see this, note that the set of PN-bridges of length $n$ is invariant (up to translation) under reflection through the real axis. But now if $\tilde\gamma$ denotes the result of reflecting a PN-bridge $\gamma$ through the real axis, then $\mathbf x(\gamma_n) - \mathbf x(\gamma_0) = -(\mathbf x(\tilde\gamma_n) - \mathbf x(\tilde\gamma_0))$.

Next, we define the concatenation of two x-irreducible P-bridges $\gamma$ and $\gamma'$ as follows: if $\gamma\in{\rm xSAPP}$ then we simply join them so that $\gamma_{|\gamma|} = \gamma'_0$; while if $\gamma\in{\rm xSAPN}$ then we join $\gamma$ to the walk $\bar{\gamma}'$ obtained by reflecting $\gamma'$ through the imaginary axis. Clearly, any P-bridge can then be decomposed uniquely into a sequence of x-irreducible P-bridges or reflections of P-bridges. We can thus adapt Kesten's relation for irreducible bridges~\cite{Kesten1963Number} to x-irreducible P-bridges, and we obtain
\[\sum_{\gamma\in{\rm xSAP}} \x^{|\gamma|} = \sum_{\gamma\in{\rm xSAPP}\,\cup\,{\rm xSAPN}}\x^{|\gamma|} = 1.\]
So we can define a probability measure $\mathbb{P}_{\rm xSAP}$ on x-irreducible P-bridges by $\mathbb{P}_{\rm xSAP}(\gamma) = \x^{|\gamma|}$. It then follows from~\eqref{eqn:Ex_NP_PN} that
\[\mathbb{E}_{\rm xSAP}(\mathbf x(\gamma_{|\gamma|})\,|\,\gamma\in{\rm xSAPN})= 0,\]
that is, the expected difference between the $\mathbf x$-coordinates of the start and end of an x-irreducible PN-bridge is 0. By symmetry, the same applies to x-irreducible NP-bridges.

Now by Proposition~\ref{prop:H_W_finite}, we have $\EiSAPP(\width(\gamma))<\infty$. Since xSAPP $\subset$ iSAPP, it follows that $\mathbb{E}_{\width}:=\mathbb{E}_{\rm xSAP}(\width(\gamma)\,|\,\gamma\in{\rm xSAPP})<\infty$.

Lastly, we define $\mathbb{P}_{{\rm R}}:= \mathbb{P}_{\rm xSAP}(\gamma\in{\rm xSAPP})$; that is, $\mathbb{P}_{{\rm R}}$ is the probability that a random $\gamma\in{\rm xSAP}$ will end on a positive mid-edge.

Now any irreducible PP-bridge is either x-irreducible, or can be decomposed uniquely into a sequence of concatenated x-irreducible P-bridges or N-bridges. We thus have a new way to generate a random irreducible PP-bridge $\gamma$, as follows:
\begin{itemize}
\item Take a random sample $\gamma^{(1)}$ of $\mathbb{P}_{\rm xSAP}$. If $\gamma^{(1)}\in{\rm xSAPP}$, then $\gamma=\gamma^{(1)}$.
\item If instead $\gamma^{(1)}\in{\rm xSAPN}$, take another sample $\gamma^{(2)}$ of $\mathbb{P}_{\rm xSAP}$, and concatenate $\gamma^{(1)}$ and the reflection $\bar{\gamma}^{(2)}$ of $\gamma^{(2)}$. (Note that $\bar{\gamma}^{(2)}$ is an N-bridge.) If the resulting walk $\gamma^{(1)}\circ\bar{\gamma}^{(2)}\in{\iSAPP}$ (i.e.~if it ends on a positive mid-edge), then $\gamma=\gamma^{(1)}\circ\bar{\gamma}^{(2)}$.
\item Otherwise, continue in this fashion by repeatedly sampling $\mathbb{P}_{\rm xSAP}$ and attaching the reflection to the current walk, and stop when the walk ends on a positive mid-edge.
\end{itemize}
Since all the samples $\gamma^{(1)},\gamma^{(2)},\ldots$ are independent and the sum of their lengths is the length of $\gamma$, the probability distribution of walks obtained in this way is
\begin{align*}
\mathbb{P}_{\rm xSAP}(\gamma^{(1)},\gamma^{(2)},\ldots) = \mathbb{P}_{\rm xSAP}(\gamma^{(1)}) \mathbb{P}_{\rm xSAP}(\gamma^{(2)})\cdots
= \x^{|\gamma^{(1)}|}\x^{|\gamma^{(2)}|}\cdots
= \x^{|\gamma|}
= \PiSAPP(\gamma).
\end{align*}
Now
\begin{multline}\label{eqn:expectedx_decompose}
\EiSAPP(\mathbf x(\gamma_{|\gamma|})) = \sum_{n=1}^\infty\PiSAPP(\gamma \text{ decomposes into } n \text{ x-irreducible P-bridges})\\\cdot\EiSAPP(\mathbf x(\gamma_{|\gamma|})\,|\,\gamma \text{ decomposes into } n \text{ x-irreducible P-bridges}).
\end{multline}
We have
\begin{multline*}
\PiSAPP(\gamma \text{ decomposes into } n \text{ x-irreducible P-bridges})\\
\begin{split}
&= \begin{cases}\mathbb{P}_{\rm xSAP}(\gamma^{(1)}\in{\rm xSAPP}) & \text{if } n=1\\ \mathbb{P}_{\rm xSAP}(\gamma^{(1)}\in{\rm xSAPN},\gamma^{(i)}\in{\rm xSAPP} \text{ for } 2\leq i\leq n-1, \gamma^{(n)}\in{\rm xSAPN}) & \text{if } n\geq 2\end{cases}\\
&= \begin{cases}\mathbb{P}_{\rm R} & \text{if } n=1\\ \mathbb{P}_{\rm R}^{n-2}(1-\mathbb{P}_{\rm R})^2 &\text{if }n\geq 2,\end{cases}
\end{split}
\end{multline*}
and
\begin{align*}
\EiSAPP(&\mathbf x(\gamma_{|\gamma|})\,|\,\gamma \text{ decomposes into } n \text{ x-irreducible P-bridges}) \\
&= \begin{cases}\mathbb{E}_{\rm xSAP}(\mathbf x(\gamma^{(1)}_{|\gamma^{(1)}|})\,|\,\gamma^{(1)}\in{\rm xSAPP}) & \text{if }n=1\\
\mathbb{E}_{\rm xSAP}(\mathbf x(\gamma^{(1)}_{|\gamma^{(1)}|})\,|\,\gamma^{(1)}\in{\rm xSAPN}) - \mathbb{E}_{\rm xSAP}(\mathbf x(\gamma^{(2)}_{|\gamma^{(2)}|})\,|\,\gamma^{(2)}\in{\rm xSAPN}) & \text{if } n=2\\
 \begin{gathered}\hspace{-2pt}\mathbb{E}_{\rm xSAP}(\mathbf x(\gamma^{(1)}_{|\gamma^{(1)}|})\,|\,\gamma^{(1)}\in{\rm xSAPN})\hspace{5cm} \\ - \sum_{i=2}^{n-1}\mathbb{E}_{\rm xSAP}(\mathbf x(\gamma^{(i)}_{|\gamma^{(i)}|})\,|\,\gamma^{(i)}\in{\rm xSAPP}) \\\shoveright{- \mathbb{E}_{\rm xSAP}(\mathbf x(\gamma^{(n)}_{|\gamma^{(n)}|})\,|\,\gamma^{(n)}\in{\rm xSAPN})}\end{gathered} & \text{if }n\geq 3
\end{cases}\\
\intertext{(The second term in the $n=2$ case and the second and third terms in the $n\geq3$ case are negative because every x-irreducible P-bridge in the decomposition, except for the first piece, gets reflected.)}
&= \begin{cases}\mathbb{E}_{\width} &\text{if }n=1\\ 0 &\text{if } n=2 \\ -\displaystyle\sum_{i=2}^{n-1}\mathbb{E}_{\width} &\text{if }n\geq3\end{cases}\\
&= -(n-2)\mathbb{E}_{\width}.
\end{align*}
So returning to~\eqref{eqn:expectedx_decompose}, we have
\begin{align*}
\EiSAPP(\mathbf x(\gamma_{|\gamma|})) &= \mathbb{P}_{\rm R}\mathbb{E}_{\width} - \sum_{n=2}^\infty (n-2)\mathbb{P}_{\rm R}^{n-2}(1-\mathbb{P}_{\rm R})^2\mathbb{E}_{\width}\\
&= \mathbb{P}_{\rm R}\mathbb{E}_{\width} - (1-\mathbb{P}_{\rm R})^2\mathbb{E}_{\width}(\mathbb{P}_{\rm R} + 2 \mathbb{P}_{\rm R}^2 + 3 \mathbb{P}_{\rm R}^3+\cdots)\\
&= \mathbb{P}_{\rm R}\mathbb{E}_{\width} - (1-\mathbb{P}_{\rm R})^2\mathbb{E}_{\width}\cdot\frac{\mathbb{P}_{\rm R}}{(1-\mathbb{P}_{\rm R})^2}\\
&=0.
\end{align*}
\end{proof}

\begin{proof}[Proof of Proposition~\ref{prop:diamond_positive_density}]
We begin by proving $\PZiSAPP(\gamma_0\in {\mathbf D}_\gamma)>0$. By Lemma~\ref{lem:expected_x_0}, we have $\EiSAPP(\mathbf x(\gamma_{|\gamma|}))=0$. The law of large numbers thus implies that, $\PNiSAPP$-almost surely, $\mathbf x(\gamma_{\bfr_n})/n\to0$. Since the expected width of irreducible PP-bridges is finite, a classical use of the Borel-Cantelli Lemma shows that $\width(\gamma_{[\bfr_n,\bfr_{n+1}]})/n\to0$ almost surely. Thus
\[\frac{1}{n}(|\mathbf x(\gamma_{\bfr_n})| + \width(\gamma_{[\bfr_n,\bfr_{n+1}]}))\to0 \qquad\text{a.s.}\]
Since
\[\width(\gamma_{[0,\bfr_n]}) \leq 2\max\big\{\textstyle{\frac{2}{3}}|\mathbf x(\gamma_{\bfr_k})| + \width(\gamma_{[\bfr_k,\bfr_{k+1}]}),0\leq k\leq n-1\big\},\]
we find that, $\PNiSAPP$-almost surely, $\width(\gamma_{[0,\bfr_n]})/n\to0$.

On the other hand, applying the law of large numbers to $\mathbf y(\gamma_{\bfr_n})$, we obtain that $\PNiSAPP$-almost surely, $\mathbf y(\gamma_{\bfr_n})/n\to\frac{\sqrt{3}}{2}\EiSAPP(\height(\gamma))>0$.

Now define
\[I(\gamma) := \inf_k\left(\mathbf y(\gamma_k)+\textstyle{\frac{\sqrt{3}}{4}}-\sqrt{3}\left|\mathbf x(\gamma_k)+\frac{1}{4}\right|\right),\]
and note that for an infinite bridge $\gamma=(\gamma_0,\gamma_1,\ldots)$, the origin $\gamma_0$ is a diamond point if and only if $I(\gamma)\geq0$.
Furthermore, by the discussion at the start of this proof, $I(\gamma)$ is finite (that is, $I(\gamma)>-\infty$) $\PNiSAPP$-almost surely. Let $K\in\mathbb{N}$ be such that $p_K:=\PNiSAPP(I(\gamma)\geq-K)>0$. We are going to show that
\begin{equation}\label{eqn:show_p0>0}
p_0 \geq \x^{2K}p_K >0.
\end{equation}
To prove~\eqref{eqn:show_p0>0}, consider an experiment under which the law $\PNiSAPP$ is constructed by first concatenating $K$ independent samples of $\PiSAPP$ (starting from $a_0$) and then an independent sample $\gamma'$ of $\PNiSAPP$. If each of the $K$ samples happens to be a walk of length $2$ going from $a_0$ to $a_0+\sqrt{3}\ii$ and $I(\gamma')\geq-K$, then $\gamma$ satisfies $I(\gamma)\geq0$. The probability that the $i^{\rm th}$ sample of $\PiSAPP$ is a walk of length $2$ going from $a_0$ to $a_0+\sqrt{3}\ii$ is $\x^2$. Thus, the experiment behaves as described with probability $\x^{2K}p_K$, and we obtain~\eqref{eqn:show_p0>0}, and hence that $\PNiSAPP(\gamma_0\in{\mathbf D}_\gamma)>0$.

Using Property $({\rm P}_3)$ of Proposition~\ref{prop:ergodic_etc}, we deduce that
\[\delta:=\PZiSAPP(\gamma_0\in{\mathbf D}_\gamma) = \left(\PNiSAPP(\gamma_0\in{\mathbf D}_\gamma)\right)^2 >0.\]
The shift $\tau$ being ergodic (cf.~Property $({\rm P}_2)$ of Proposition~\ref{prop:ergodic_etc}), the ergodic theorem, applied to $\mathbbm{1}_{\gamma_0\in{\mathbf D}_\gamma}$, gives
\[\PZiSAPP\left(\lim_{n\to\infty}\frac{|{\mathbf D}_\gamma\cap\{0,\ldots,\bfr_n(\gamma)\}|}{n}=\delta\right)=1.\]
Let $\gamma$ be a bi-infinite bridge and denote $\gamma^+ = \gamma_{[0,\infty)}$. Then for $n\geq0$, $\bfr_n(\gamma) = \bfr_n(\gamma^+)$, and
\[{\mathbf D}_\gamma\cap\{0,\ldots,\bfr_n(\gamma)\} = {\mathbf D}_\gamma\cap\{0,\ldots,\bfr_n(\gamma^+)\} \subset {\mathbf D}_{\gamma^+}\cap\{0,\ldots,\bfr_n(\gamma^+)\}\]
since all diamond points of $\gamma$ are diamond points of $\gamma^+$. This implies that
\begin{align*}
\PNiSAPP\left(\liminf_{n\to\infty}\frac{|{\mathbf D}_\gamma\cap\{0,\ldots,\bfr_n(\gamma)\}|}{n}\geq\delta\right)
&=\PZiSAPP\left(\liminf_{n\to\infty}\frac{|{\mathbf D}_{\gamma^+}\cap\{0,\ldots,\bfr_n(\gamma)\}|}{n}\geq\delta\right)\\
&\geq \PZiSAPP\left(\liminf_{n\to\infty}\frac{|{\mathbf D}_\gamma\cap\{0,\ldots,\bfr_n(\gamma)\}|}{n}\geq\delta\right)\\
&=1.
\end{align*}
\end{proof}

We are now almost ready to complete the proof of the main result of this appendix. The final lemma here serves the same role as Lemma 17 in~\cite{Beaton2013Critical}. We henceforth assume (for a contradiction) that $\EiSAPP(\height(\gamma))<\infty$, and thus let $\nu>\EiSAPP(\height(\gamma))$. Also let $0<\epsilon<\delta/20$, where $\delta$ satisfies Proposition~\ref{prop:diamond_positive_density}.

Let $\Omega^+$ denote the set of semi-infinite walks in the upper half-plane. That is, $\phi=(\phi_0, \phi_1, \ldots) \in \Omega^+$ if and only if  $y(\phi_i) > 0$ for $i > 0$. For $\phi \in \Omega^+$ and $\gamma$ a  finite PP-bridge, we write $\gamma\triangleleft \phi$ if  $\phi_{[0,|\gamma|]}=\gamma$ and $\phi_{| \gamma |}$ is a renewal
 point of $\phi$. Note that 
\begin{equation}\label{eqn:eqtrileft} 
\x ^{|\gamma|}=\PNiSAPP(\phi\in \Omega^+: \gamma\triangleleft \phi) \, .
\end{equation}

Let  $\overline{\text{SAPP}}_n$ denote the set of finite PP-bridges $\gamma$ with exactly $n+1$ renewal points (meaning that $\mathbf r_n(\gamma)=|\gamma|$) such that 
\begin{itemize}
\item[$({\rm C}_1)$] $\height(\gamma) \le  \nu  n$,
\item[$({\rm C}_2)$] $|\mathbf{D}_\gamma|\ge \delta n/2$.
\end{itemize}
Let us define $\overline{\text{SAPP}}_n^+=\{\phi\in \Omega^+ :\exists \gamma\in \overline{\text{SAPP}}_n\text{ such that }\gamma\triangleleft \phi\}$. That is, the prefix of $\phi$ consisting of its $n$ first irreducible bridges satisfies $({\rm C}_1)$ and $({\rm C}_2)$. It follows from~\eqref{eqn:eqtrileft} that
\begin{equation}\label{eqn:sarequiv}
\PNiSAPP\big( \overline{\text{SAPP}}_n^+\big)=\sum_{\gamma\in\overline{\text{SAPP}}_n}\x ^{|\gamma|} \, .
\end{equation}
The proof of the following is identical to that of Lemma~17 in~\cite{Beaton2013Critical}.

\begin{lem}\label{lem:PrN_barSAPPplus_to1}
Under the above assumptions we have, as $n\to\infty$,
\[\PNiSAPP(\overline{\rm{SAPP}}_n^+)\to1.\]
\end{lem}

\begin{proof}[Proof of Theorem~\ref{thm:appendix_thm}]
This proof is essentially identical to that of Theorem 10 in~\cite{Beaton2013Critical}, and we direct interested readers to that article for further details. We we will mention just one important detail, which is the definition of the $\mathsf{StickBreak}$ operation. (See Figure~\ref{fig:stickbreak} for an illustration.) For a finite PP-bridge $\gamma$ with distinct diamond points at indices $0<{\bf d}_i<{\bf d}_j<|\gamma|$, we define
\[\mathsf{StickBreak}_{i,j}(\gamma) = \gamma_{[0,{\bf d}_i]}\circ \mathbbm{r}\circ \rho(\gamma_{[{\bf d}_i,{\bf d}_j]})\circ\mathbbm{l}\circ\gamma_{[{\bf d}_j,|\gamma|]},\]
where $\circ$ stands for concatenation, $\rho$ is the clockwise rotation through angle $\pi/3$, $\mathbbm{r}$ is a single right turn and $\mathbbm{l}$ is a single left turn. The definition of diamond points implies that $\mathsf{StickBreak}_{i,j}(\gamma)$ is not only self-avoiding, but also a PP-bridge.

The essential idea of the proof is that by Proposition~\ref{prop:diamond_positive_density}, there will always be available diamond points in a walk on which to perform the $\mathsf{StickBreak}$ operation. But performing the $\mathsf{StickBreak}$ operation increases the width of a walk, and this ultimately contradicts Proposition~\ref{prop:H_W_finite}. Thus we are forced to conclude $\EiSAPP(\height(\gamma))=\infty$, and then by Lemma~\ref{lem:RT_expectedheight} we have $\lim_{T\to\infty}\PPT(\x)=0$.
\end{proof}

\bibliographystyle{amsplain}
\bibliography{rotated}

\end{document}